\documentclass[11pt,a4paper]{article}

\usepackage[truedimen,margin=27mm]{geometry} 

\usepackage{mathrsfs}
\usepackage{amssymb}
\usepackage{amsmath}
\usepackage{ascmac}
\usepackage{amsthm}
\usepackage[dvipdfmx]{graphicx}
\usepackage{natbib}
\usepackage{color}
\usepackage{setspace}
\usepackage{url}
\usepackage{here}
\usepackage{comment,booktabs}
\usepackage{lscape}
\usepackage{algorithm, algorithmicx, algpseudocode}
\usepackage{multirow}
\usepackage{multicol}
\usepackage{bm}
\usepackage{titlesec}
\titleformat*{\section}{\large\bfseries}
\titleformat*{\subsection}{\it}

\newtheorem{definition}{Definition}
\newtheorem{theorem}{Theorem}

\newtheorem{proposition}{Proposition}

\usepackage{authblk}
\title{{\bf Robust Bayesian graphical modeling using $\gamma$-divergence}}
\author[1]{Takahiro Onizuka}
\author[2]{Shintaro Hashimoto} 

\affil[1]{Graduate School of Social Sciences, Chiba University, Japan}
\affil[2]{Department of Mathematics, Hiroshima University, Japan}
\date{}

\begin{document}

\maketitle
\doublespacing

\begin{abstract}
    Gaussian graphical model is one of the powerful tools to analyze conditional independence between two variables for multivariate Gaussian-distributed observations. When the dimension of data is moderate or high, penalized likelihood methods such as the graphical lasso are useful to detect significant conditional independence structures. However, the estimates are affected by outliers due to the Gaussian assumption. This paper proposes a novel robust posterior distribution for inference of Gaussian graphical models using the $\gamma$-divergence which is one of the robust divergences. In particular, we focus on the Bayesian graphical lasso by assuming the Laplace-type prior for elements of the inverse covariance matrix. The proposed posterior distribution matches its maximum a posteriori estimate with the minimum $\gamma$-divergence estimate provided by the frequentist penalized method. We show that the proposed method satisfies the posterior robustness which is a kind of measure of robustness in Bayesian analysis. The property means that the information of outliers is automatically ignored in the posterior distribution as long as the outliers are extremely large. A sufficient condition for the posterior propriety of the proposed posterior distribution is also derived. Furthermore, an efficient posterior computation algorithm via the weighted Bayesian bootstrap method is proposed. The performance of the proposed method is illustrated through simulation studies and real data analysis.
\end{abstract}


\section{Introduction\label{sec:1}}

Estimating the dependence structure between variables is an important issue in multivariate analysis. Let $\bm{y}_1,\dots, \bm{y}_n$ be a sequence of independent, identically distributed random vectors according to the $p$-dimensional multivariate Gaussian distribution: 
\begin{align}\label{Gaussian-model}
\bm{y}_i\sim \mathcal{N}_p(\bm{0},\mathbf{\Omega}^{-1}), \quad i=1,\dots,n,
\end{align}
where $\bm{y}_i=(y_{i1},\dots,y_{ip})^{\top}\in\mathbb{R}^p$ for $i=1,\dots,n$, and 
$\mathbf{\Omega}=(\omega_{ij})\in\mathbb{R}^{p\times p}$ is a precision matrix defined by the inverse covariance matrix. The estimation of the precision matrix $\mathbf{\Omega}$ under the Gaussian assumption is called the Gaussian graphical model \citep{lauritzen1996graphical, whittaker2009graphical}. We also put $\mathbf{Y}=(\bm{y}_1,\dots,\bm{y}_n)^{\top} \in \mathbb{R}^{n\times p}$. Estimating the conditional dependence structures between variables corresponds to estimating whether  an off-diagonal element of the precision matrix is zero or not. 
To deal with the sparsity of the precision matrix, the penalized approach such as the graphical lasso has been considered (e.g. \cite{yuan2007model,friedman2008sparse}). The graphical lasso estimate is defined by minimizing the penalized log-likelihood
\begin{equation*}
-\log|\mathbf{\Omega}|+\mathrm{tr}(\mathbf{S}\mathbf{\Omega})+\rho\|\mathbf{\Omega}\|_1 
\end{equation*}
over the space of positive definite matrices $M^+$, where $\mathbf{S}=\mathbf{Y}^{\top}\mathbf{Y}/n$ is the sample covariance matrix, $\rho\ge 0$ is a tuning parameter, and $\|\mathbf{\Omega}\|_1=\sum_{1\le i, j\le p} |\omega_{ij}|$. \cite{friedman2008sparse} proposed an efficient optimization algorithm, which guarantees symmetry and positive definiteness of $\mathbf{\Omega}$. Although the method provides point estimates quickly, we cannot conduct a full probabilistic inference for the parameter of interest $\mathbf{\Omega}$ in graphical models.   

However, the Bayesian approach is useful for quantifying the uncertainty of the parameter of interest. 
As a Bayesian alternative to the graphical lasso, \cite{wang2012bayesian} proposed a Bayesian graphical lasso model defined by
\begin{align*}
\bm{y}_i\mid \mathbf{\Omega} & \sim \mathcal{N}_p(\bm{0}, \mathbf{\Omega}^{-1}), \quad i=1,\dots,n\\
\mathbf{\Omega}\mid \lambda & \sim C_{\lambda}^{-1}\prod_{i<j} \mathrm{Lap}(w_{ij}\mid \lambda)\prod_{i=1}^p \mathrm{Exp}\left(w_{ii}\mid \frac{\lambda}{2}\right)1_{\mathbf{\Omega}\in M^+},
\end{align*}
where $C_{\lambda}=\int \prod_{i<j} \mathrm{Lap}(w_{ij}\mid \lambda)\prod_{i=1}^p \mathrm{Exp}\left(w_{ii}\mid \frac{\lambda}{2}\right)1_{\mathbf{\Omega}\in M^+}d\mathbf{\Omega}$ is the normalizing constant of the prior density, and $\lambda>0$ is a prior hyper-parameter that plays the same role as the penalty parameter in the original graphical lasso. $\mathrm{Lap}(\cdot\mid \lambda)$ represents the probability density function of the Laplace (or double-exponential) distribution defined by $\mathrm{Lap}(x\mid \lambda)=(\lambda/2)\exp(-\lambda |x|)$, and $\mathrm{Exp}(\cdot\mid \lambda)$ represents the probability density function of the exponential distribution defined by $\mathrm{Exp}(x\mid \lambda)=\lambda \exp(-\lambda x) 1_{x>0}$. The prior for $\mathbf{\Omega} \mid \lambda$ is called the graphical lasso prior in \cite{wang2012bayesian}. The Bayesian approach enables us to quantify the uncertainty through the posterior distribution. \cite{wang2012bayesian} also proposed an efficient block Gibbs sampler to obtain the posterior sample, while there is a drawback that the estimate of $\omega_{ij}$ for $i\ne j$ cannot be exactly zero. Hence, we may need to consider a criterion to specify dependent structures. 

However, these Gaussian graphical models cannot lead to suitable estimates if the data involve outliers or data generating distribution is heavily-tailed. Although one of the remedies is to use a heavy-tailed multivariate distribution such as the multivariate $t$-distribution (see e.g., \cite{finegold2011robust}), as pointed out by \cite{hirose2017robust}, the heavy-tailed distribution generates both large and moderate outliers. Hence, the variance of the estimator tends to be large because a heavy-tail distribution often produces a small Fisher information. In fact, we show in Section \ref{sec:4} that assuming the multivariate $t$-distribution as the likelihood function leads to undesirable results for the estimation of the inverse covariance matrix. To overcome these issues, we consider the robust divergence to estimate Gaussian graphical models in the presence of outliers. In a frequentist perspective, \cite{hirose2017robust} proposed the $\gamma$-lasso, which is a robust estimation method of the inverse covariance matrix based on the $\gamma$-divergence \citep{fujisawa2008robust}. Although the density-power divergence \citep{basu1998robust} is often used, it is known that the method does not work very well for the estimation of the variance or scale parameter. Robust Bayesian modeling based on the $\gamma$-divergence has also been developed in recent years. For example, \cite{hashimoto2020robust} proposed a robust and sparse Bayesian linear regression model, and \cite{momozaki2023robustness} considered a robust ordinal response model through the $\gamma$-divergence.  

In this paper, we propose a robust Bayesian graphical lasso based on the $\gamma$-divergence. Although the proposed method combines the $\gamma$-lasso by \cite{hirose2017robust} with the Bayesian graphical lasso by \cite{wang2012bayesian}, we prove the robustness property called the posterior robustness for the proposed method. 
The posterior robustness is one of the robustness properties of posterior distributions. For the posterior distribution equipped with the property, the information of outliers is automatically rejected in the posterior distribution as long as the outliers are extremely large. Recently, some researchers studied sufficient conditions for the posterior robustness using log-regularly varying tailed probability distributions for data \citep{desgagne2015robustness, desgagne2019bayesian, hamura2022log}. However, posterior robustness for divergence-based robust methods has not been much developed. To show posterior robustness under the $\gamma$-divergence, we introduce a new robust posterior distribution whose maximum a posteriori estimate matches the estimate by $\gamma$-lasso \citep{hirose2017robust}. Since the proposed posterior distribution is a synthetic one, we provide a sufficient condition for the posterior propriety in the use of the Bayesian lasso-type priors. We also show that the other popular approaches do not satisfy the posterior robustness for the estimation of the inverse covariance matrix. We illustrate the performance of the proposed method through some numerical experiments and apply the proposed method to the analysis of gene expression data. 

The remainder of the paper is structured as follows. In Section \ref{sec:2}, a new robust posterior distribution based on the $\gamma$-divergence is proposed, and some theoretical properties and an efficient posterior computation algorithm are also presented. In Section \ref{sec:3}, we discuss robustness properties for other posterior distributions. Numerical experiments and the real data example are shown in Sections \ref{sec:4} and \ref{sec:5}, respectively. R code implementing the proposed methods is available in the GitHub repository (URL: \url{https://github.com/Takahiro-Onizuka/RBGGM-gamma}). Additional information on the proposed algorithms and numerical experiments is provided in the Supplementary Material.

\section{Robust Bayesian graphical models}\label{sec:2}

In this section, we present our main proposal and show some theoretical properties of the proposed model. Furthermore, we provide an efficient and scalable posterior computation algorithm using the weighted Bayesian bootstrap.

\subsection{Robust graphical lasso via $\gamma$-divergence}\label{sec:2.1}

It is well-known that parameter estimation under the Gaussian likelihood is affected by outliers. Studies on robust parameter estimation have a long history and many useful methods have been proposed in the literature. One of the methods is to use the heavy-tailed probability distribution instead of the Gaussian distribution. For example, \cite{finegold2011robust} proposed a robust graphical lasso based on the multivariate $t$-distribution. However, constructing such a distribution equipped with desirable robustness properties is not straightforward, especially in multivariate cases. As a more versatile approach, divergence-based or weighted likelihood methods have been developed in the last two decades (see, e.g. \cite{basu1998robust}). In this paper, we focus on a robust divergence called $\gamma$-divergence \citep{fujisawa2008robust}. \cite{hirose2017robust} considered a robust Gaussian graphical modeling based on the $\gamma$-divergence. 
We provide a brief introduction of the $\gamma$-divergence.
The $\gamma$-divergence between a data generating process $g(\bm{y})$ and probability density function $f_{\bm{\theta}}=f(\bm{y}\mid \bm{\theta})$ is defined by
\begin{align*}
d_{\gamma}(g,f_{\bm{\theta}}) = \frac{1}{\gamma(1+\gamma)} \log \int g(\bm{y})^\gamma d\bm{y} -\frac{1}{\gamma} \log \int g(\bm{y})f(\bm{y}\mid \bm{\theta})^\gamma d\bm{y} + \frac{1}{1+\gamma} \log \int f(\bm{y}\mid \bm{\theta})^{1+\gamma} d\bm{y},
\end{align*}
where $\gamma>0$ is a tuning parameter to control the balance between efficiency and robustness. Following \cite{fujisawa2008robust}, we now explain that the minimum $\gamma$-divergence estimate has a strong robustness property.
It is often assumed that the data-generating distribution is contaminated as $g(\bm{y})=(1-\varepsilon) f(\bm{y}) + \varepsilon \delta(\bm{y})$, where $f(\bm{y})$ is the underlying target distribution, $\delta(\bm{y})$ is a contamination distribution and $\varepsilon>0$ is the contamination ratio. Furthermore, we assume that $\nu(\bm{\theta},\gamma)=\int \delta(\bm{y}) f_{\bm{\theta}}^{\gamma}(\bm{y}) d\bm{y} \approx 0$ for $\gamma>0$. Note that we do not assume that the contamination ratio is small. Then the corresponding $\gamma$-divergence is expressed by
\begin{align*}
d_{\gamma} (g,f_{\bm{\theta}}) &= -\frac{1}{\gamma} \log \int g(\bm{y}) f_{\bm{\theta}}^{\gamma}(\bm{y}) d\bm{y} +\frac{1}{1+\gamma} \log\int f_{\bm{\theta}}^{1+\gamma}(\bm{y}) d\bm{y} \\
& \approx -\frac{1}{\gamma} \log \left\{ (1-\varepsilon) \int f(\bm{y})\cdot f_{\bm{\theta}}^{\gamma}(\bm{y}) d\bm{y} +0 \right\}  +\frac{1}{1+\gamma} \log \int f_{\bm{\theta}}^{1+\gamma}(\bm{y}) d\bm{y}\\
&= d_{\gamma}(f,f_{\bm{\theta}}) -\frac{1}{\gamma} \log (1-\varepsilon).
\end{align*}
Since the term $\frac{1}{\gamma} \log (1-\varepsilon)$ does not depend on $\bm{\theta}$, we have $\arg\min_{\bm{\theta}} d_{\gamma} (g,f_{\bm{\theta}}) \approx \arg \min_{\bm{\theta}} d_{\gamma}(f,f_{\bm{\theta}})$ under the assumption $\nu(\bm{\theta},\gamma) \approx 0$. 
For example, let $\delta (y) = \mathcal{N}(y\mid \alpha,1)$ and $f_{\theta}(y)=\mathcal{N}(y \mid \theta,1)$. Then $\nu(\theta,\gamma)=c_{1,\gamma} \exp\{- c_{2,\gamma} (\alpha -\theta)^2\}\approx 0$ for large $\alpha$, where $c_{1,\gamma}$ and $c_{2,\gamma}$ are constant numbers. In this paper, we show that a similar property also holds in the posterior distribution (not only point estimate) if the posterior distribution is properly defined. 

If we use a large $\gamma$, it is known that the corresponding estimate becomes more robust \cite{fujisawa2008robust}. However, the efficiency of the estimator decreases. In Figure 3 and Table 5 of the article \cite{nakagawa2020default}, they discussed the trade-off in terms of the asymptotic relative efficiency, and it was observed that larger $\gamma$ is not always better. Since $\gamma$ determines the shape of the discrepancy/loss function $d_{\gamma}(g,f_{\theta})$ and the loss function is usually selected depending on the purpose of analysis, we often select $\gamma$ as a fixed small positive value (see, e.g., \cite{hirose2017robust, hashimoto2020robust}). 
The selection method of $\gamma>0$ has not been clear in general, but some strategies have been developed in recent years (see, e.g., \cite{yonekura2023adaptation}). Although \cite{yonekura2023adaptation} proposed a data-dependent selection of $\gamma$ in the Bayesian framework using the sequential Monte Carlo method, their simulation mainly dealt with a univariate probability distribution, and the computation cost is not low even in this case. 
For this reason, we suggest a fixed $\gamma$ as a small positive value (e.g. $\gamma=0.05$ or $0.1$) in the proposed method. 

Since $g(\bm{y})$ is unknown in practice, the minimum $\gamma$-divergence estimate is obtain by solving the following optimization problem: 
\begin{align}\label{gamma-opt}
\min_{\bm{\theta}} \left\{ -\frac{1}{\gamma}\log\left\{\frac{1}{n}\sum_{i=1}^n f(\bm{y}_i\mid \bm{\theta})^{\gamma}\right\}+\frac{1}{1+\gamma}\log\int f(\bm{y}\mid\bm{\theta})^{1+\gamma}d\bm{y}\right\}.
\end{align}
The objective function in \eqref{gamma-opt} is also called the negative $\gamma$-likelihood function. Combining the $\gamma$-likelihood with $L_1$-penalty, \cite{hirose2017robust} proposed a robust and sparse graphical lasso model. In the following sections, we consider robust graphical lasso models from a Bayesian perspective.

\subsection{MAP $\gamma$-posterior distribution} \label{sec:2.2}

Robust Bayesian modeling via the $\gamma$-divergence has been developed in recent years (see \cite{nakagawa2020robust, hashimoto2020robust}). Although existing studies deal with robust Bayesian inference for univariate observations and linear regression models, we here focus on robust Bayesian inference for a precision matrix in multivariate Gaussian-distributed observations. To this end, we introduce a robust posterior distribution based on the $\gamma$-divergence. 
\cite{hashimoto2020robust} proposed a synthetic posterior distribution based on the $\gamma$-divergence that converges to the standard posterior as $\gamma\to 0$, and mainly considered a robust estimation of sparse linear regression models. \cite{nakagawa2020robust} also proposed a posterior distribution based on the monotone transformed $\gamma$-divergence. However, we introduce another type of posterior distribution that focuses on matching the maximum a posteriori (MAP) estimates with the corresponding frequentist optimal solution (e.g. \cite{park2008bayesian,wang2012bayesian}). 
In general, a objective function based on the $\gamma$-likelihood with a penalty term $\lambda\phi(\bm{\theta})$ is defined by
\begin{align}
L_{\gamma}(\bm{\theta})=-\frac{1}{\gamma}\log\left\{\frac{1}{n}\sum_{i=1}^n f(\bm{y}_i\mid \bm{\theta})^{\gamma}\right\}+\frac{1}{1+\gamma}\log\int f(\bm{y}\mid \bm{\theta})^{1+\gamma}d\bm{y}+\lambda\phi(\bm{\theta}). \label{objective-func}
\end{align}
Note that the first two terms in \eqref{objective-func} are the same as \cite{hirose2017robust}. The penalized objective function has a natural counterpart as a Bayesian posterior distribution as follows: 
\begin{align*}
\pi(\bm{\theta}\mid \mathbf{Y}) \propto \exp\left(\frac{1}{\gamma}\log\left\{\frac{1}{n}\sum_{i=1}^n f(\bm{y}_i\mid \bm{\theta})^{\gamma}\right\}-\frac{1}{1+\gamma}\log\int f(\bm{y}\mid \bm{\theta})^{1+\gamma}d \bm{y}\right)\exp\left(-\lambda\phi(\bm{\theta})\right),
\end{align*}
where the first and second terms are interpreted as likelihood and prior density functions, respectively. The MAP estimate based on the posterior distribution is equal to the minimizer of the penalized objective function \eqref{objective-func}. On the other hand, the MAP estimates based on the posteriors of \cite{hashimoto2020robust} and \cite{nakagawa2020robust} do not match the minimizer of \eqref{objective-func}. 
To avoid this problem, we define a new robust posterior distribution as
\begin{align}\label{MAP-posterior-general}
\pi_{\gamma}(\bm{\theta}\mid \mathbf{Y})&\propto \exp\left[\frac{1}{\gamma}\log\left(\frac{1}{n}\sum_{i=1}^n f(\bm{y}_i\mid \bm{\theta})^{\gamma}\right)-\frac{1}{1+\gamma}\log\left(\int f(\bm{y}\mid \bm{\theta})^{1+\gamma}d\bm{y}\right)\right]\pi(\bm{\theta}).
\end{align}
In this paper, we call the posterior {\it MAP $\gamma$-posterior}. The posterior is different from those of \cite{hashimoto2020robust} and \cite{nakagawa2020robust}, but there are some advantages when we consider the estimation of Gaussian graphical models; 1) the MAP estimate coincides with the frequentist solution; 2) the corresponding posterior density is easy to handle for proving theoretical properties; 3) existing frequentist optimization methods can be directly used to sample from the posterior distribution. 

Hereafter, we focus on the estimation of Gaussian graphical models defined by \eqref{Gaussian-model}, that is, the parameter is a positive definite precision matrix $\mathbf{\Omega}\in \mathbb{R}^{p\times p}$ and the density function for $\bm{y}_i$ is given by
\begin{align*}
f(\bm{y}_i\mid \mathbf{\Omega})=(2\pi)^{-p/2} |\mathbf{\Omega}|^{1/2}\exp(-\bm{y}_i^{\top}\mathbf{\Omega} \bm{y}_i/2).
\end{align*}
Then the corresponding MAP $\gamma$-posterior \eqref{MAP-posterior-general} is given by
\begin{align}\label{MAP-posterior}
\pi_{\gamma}(\bm{\Omega}\mid \mathbf{Y})= \frac{|\mathbf{\Omega}|^{1/2(1+\gamma)}\left\{\sum_{i=1}^n\exp\left(-\frac{\gamma}{2}\bm{y}_i^{\top}\mathbf{\Omega} \bm{y}_i\right)\right\}^{1/\gamma}\pi(\mathbf{\Omega})}{\int |\mathbf{\Omega}|^{1/2(1+\gamma)}\left(\sum_{i=1}^n\exp\left(-\frac{\gamma}{2}\bm{y}_i^{\top}\mathbf{\Omega} \bm{y}_i\right)\right\}^{1/\gamma}\pi(\mathbf{\Omega})d\mathbf{\Omega}}.
\end{align}

The validity of the (discrepancy-based) posterior distributions \eqref{MAP-posterior} can be explained using a concept called ``general posterior distributions" \cite{bissiri2016general}. A different point from standard Bayesian posterior distributions is whether the target parameter of inference is the parameter in the model (or likelihood) or whether the parameter minimized the discrepancy $d_{\gamma}(g,f_{\theta})$. The general Bayesian framework aims to infer that the parameter minimized the discrepancy, and \cite{bissiri2016general} showed that such types of posterior distributions have decision-theoretic validity and coherence properties under some conditions. Hence, for example, the number 95\% in the corresponding credible interval has probabilistic interpretation for the parameter that minimized the discrepancy $d_{\gamma}(g,f_{\theta})$ as long as the posterior distribution is proper. Since the coherence property does not hold for the proposed posterior distribution, sequential updating for the proposed posterior distribution is not valid. If we adopt the monotonically transformed $\gamma$-divergence as in \cite{nakagawa2020robust}, the coherence property for sequential updating of the posterior holds. However, since such a transformation makes it difficult to derive the theoretical properties of the posterior distribution, we do not consider it here. 

\subsection{Theoretical properties}\label{subsec:2.3}

We show two important theoretical results on the proposed posterior distribution \eqref{MAP-posterior}. The first is the posterior propriety, defined as follows. Let $\bm{y}_1,\dots,\bm{y}_n$ be a sequence of independent, identically distributed random variables according to the density function $f(\bm{y} \mid\bm{\theta})$, and let $\pi(\bm{\theta})$ be a prior density for $\bm{\theta}$. The posterior distribution is called proper if the normalizing constant satisfies $\int f(\bm{y}_1,\dots,\bm{y}_n\mid \bm{\theta}) \pi(\bm{\theta}) d\bm{\theta} < \infty$ (see, e.g., \cite{berger2009formal}), where $f(\bm{y}_1,\dots,\bm{y}_n\mid \bm{\theta})=\prod_{i=1}^n f(\bm{y}_i\mid \bm{\theta})$. When we assume a proper probabilistic model as a likelihood and a proper prior distribution, the posterior is proper. However, an improper probabilistic model with respect to $\bm{y}$ such as \eqref{MAP-posterior} does not always lead to a proper posterior distribution even if we assume a proper prior for $\bm{\theta}$. Hence, discussing the posterior propriety of the proposed model is an important issue when employing the proposed posterior distribution. 

The following theorem provides a sufficient condition for the posterior propriety under the proposed model \eqref{MAP-posterior}.

\begin{theorem}\label{posterior-proper}
Assume that the prior for $\mathbf{\Omega}$ is written by $\pi(\mathbf{\Omega})=\prod_{i\le j}\pi(\omega_{ij})1_{\mathbf{\Omega}\in M^+}$, and $\prod_{i< j}\pi(\omega_{ij})$ is proper. If there exists an integrable function $g(w_{ii})$ such that 
\begin{align*}
\omega_{ii}^{1/(2+2\gamma)}\pi(\omega_{ii})\le g(\omega_{ii}),\quad i=1,\dots,p,
\end{align*}
then the posterior distribution $\pi_{\gamma}(\mathbf{\Omega}\mid \mathbf{Y})$ is proper for all $\mathbf{Y}=(\bm{y}_1,\dots, \bm{y}_n)^{\top}\in \mathbb{R}^{n\times p}$. 
\end{theorem}

\begin{proof}[\textbf{\upshape Proof:}]
Since $\bm{y}\mid\mathbf{\Omega} \sim \mathcal{N}_p(\bm{0},\mathbf{\Omega}^{-1})$, we have
\begin{align*}
\int f(\bm{y}\mid\mathbf{\Omega})^{1+\gamma}d\bm{y}=(2\pi)^{-p\gamma/2}|\mathbf{\Omega}|^{\gamma/2} \left(1+\gamma\right)^{-p/2}. 
\end{align*}
The posterior density under the prior $\pi(\mathbf{\Omega})=\prod_{i<j}\pi(\omega_{ij})\prod_{i=1}^p\pi(\omega_{ii})1_{\mathbf{\Omega}\in M^+}$ is bounded by
\begin{align}
\pi_{\gamma}(\mathbf{\Omega}\mid \mathbf{Y})&\propto  |\mathbf{\Omega}|^{1/(2(1+\gamma))}\left(\sum_{i=1}^n\exp\left(-\frac{\gamma}{2}\bm{y}_i^{\top}\mathbf{\Omega} \bm{y}_i\right)\right)^{1/\gamma} \pi(\mathbf{\Omega})\notag\\
&\le C\prod_{i=1}^p \omega_{ii}^{1/(2(1+\gamma))}\pi(\mathbf{\Omega})\notag\\
&=C\prod_{i<j}\pi(\omega_{ij})\prod_{i=1}^p\pi(\omega_{ii})\omega_{ii}^{n/(2(1+\gamma))}1_{\mathbf{\Omega}\in M^+}, \label{pos-upper}
\end{align}
where $C$ is a constant and the last inequality follows from Hadamard's inequality under a positive definite matrix $\mathbf{\Omega}$. From the assumptions, there exists an integrable function $g(\omega_{ii})$ $(i=1,\dots,p)$ such that $\omega_{ii}^{1/(2+2\gamma)}\pi(\omega_{ii})\le g(\omega_{ii})$. Then the integration of the right-hand side of \eqref{pos-upper} is bounded by
\begin{align*}
\int \prod_{i<j}\pi(\omega_{ij})\prod_{i=1}^p\pi(\omega_{ii})\omega_{ii}^{n/(2(1+\gamma))}1_{\mathbf{\Omega}\in M^+}d\mathbf{\Omega}
&=\prod_{i=1}^p\int \pi(\omega_{ii})\omega_{ii}^{n/(2(1+\gamma))}1_{\mathbf{\Omega}\in M^+}d\omega_{ii}\\
&\le \prod_{i=1}^p\int g(\omega_{ii})d\omega_{ii}<\infty
\end{align*}
Therefore, the posterior distribution $\pi_{\gamma}(\mathbf{\Omega}\mid \mathbf{Y})$ is proper. 

\end{proof}

From Theorem \ref{posterior-proper}, the tail behavior of the prior for diagonal element $\omega_{ii}$ ($i=1,\dots,p$) is important for the posterior propriety, while we can use any proper prior for off-diagonal elements of $\mathbf{\Omega}$. Note that if we assume $\omega_{ii}\sim\mathrm{Exp}(\lambda/2)$ as \cite{wang2012bayesian}, then the posterior is proper, but we cannot apply improper priors such as improper uniform prior to the proposed model. Even if the prior distribution is proper, we cannot employ the Cauchy prior distribution because the distribution does not satisfy the assumption of Theorem \ref{posterior-proper}. This is a different point from \cite{li2019graphical} where they employ an improper uniform prior for the diagonal element $\omega_{ii}$ ($i=1\dots,p$). 

Next, we show that the proposed model has a desirable robustness property in the presence of outliers. Before we state the result, we introduce the definition of {\it posterior robustness} (see, e.g., \cite{desgagne2015robustness, desgagne2019bayesian, Gagnon2020, hamura2022log}), which is known as a Bayesian measure of robustness. Following \cite{desgagne2019bayesian}, we define an outlier for multivariate observations. We consider observations $\mathbf{Y}=(\bm{y}_1,\dots,\bm{y}_n)^{\top}\in \mathbb{R}^{n\times p}$, and assume that each $y_{ij}$ is expressed by
\begin{align*}
y_{ij}=
\begin{cases}
a_{ij} & (i\in\mathcal{K}),\\
a_{ij} + b_{ij}z& (i\in\mathcal{L})
\end{cases}
\end{align*}
for $a_{ij}\in\mathbb{R}$, $b_{ij}\in\mathbb{R}$, and $z>0$ for $i=1,\dots,n$ and $j=1,\dots,p$, where $\mathcal{K}$ and $\mathcal{L}$ denote the sets of indices that are non-outliers and outliers, respectively. We note that $\mathcal{K}$ and $\mathcal{L}$ are satisfied with $\mathcal{K}\cup\mathcal{L}=\{1,\dots,n\}$ and $\mathcal{K}\cap\mathcal{L}=\emptyset$. Therefore, some elements $y_{ij}$ in $\bm{y}_i$ for $i\in\mathcal{L}$ are represented by $a_{ij} + b_{ij}z$ when $b_{ij}\neq0$. If $z$ is large, then $y_{ij}$ takes a large value and the resulting vector $\bm{y}_i$ is considered an outlier. For $i\in\mathcal{K}$, $\bm{y}_i=(y_{i1},\dots,y_{ip})^\top=(a_{i1},\dots,a_{ip})^\top$. Additionally, let $\mathcal{D}=\{\bm{y}_i \mid i \in \mathcal{K}\cup \mathcal{L}\}$ be a set of all observations and let $\mathcal{D}^*=\{\bm{y}_i\mid i\in \mathcal{K}\}$ be a set of non-outlying observations. In general, the posterior robustness is defined as follows.

\begin{definition}
[Posterior robustness]\label{def1}
A proper posterior distribution $\pi(\bm{\theta}\mid \mathcal{D})$ satisfies the posterior robustness if it holds that
\begin{align*}
\quad\lim_{z\to\infty} \pi(\bm{\theta}\mid \mathcal{D}) = \pi(\bm{\theta}\mid \mathcal{D}^*).
\end{align*}
\end{definition}
We note that the convergence in Definition \ref{def1} is $L_1$ sense, that is, $\int |\pi(\bm{\theta}\mid \mathcal{D}) - \pi(\bm{\theta} \mid \mathcal{D}^*)| d\bm{\theta} \to 0$ as $z\to \infty$. Intuitively, the definition means that the information of the outliers is automatically ignored in the posterior distribution $\pi(\bm{\theta}\mid \mathcal{D})$ as long as the outliers are extremely large. In other words, the outliers and non-outliers are well-separated. The property is an analog of a redescending property in frequentist robust statistics \citep{maronna2019robust}. \cite{fujisawa2008robust} and \cite{hirose2017robust} also discussed a redescending property of the $\gamma$-divergence, while they did not give an explicit definition of outliers. We have the following result on the posterior robustness of the proposed model given by \eqref{MAP-posterior}.

\begin{theorem}\label{gamma-robust}
Assume that the posterior $\pi_{\gamma}(\mathbf{\Omega}\mid \mathbf{Y})$ is proper for all observations. Then the proposed posterior distribution \eqref{MAP-posterior} satisfies the posterior robustness.
\end{theorem}

\begin{proof}[\textbf{\upshape Proof:}]
For $\mathcal{D}$ and $\mathcal{D}^*$, the ratio of the posterior densities \eqref{MAP-posterior} is expressed by
\begin{align*}
\frac{\pi_{\gamma}(\mathbf{\Omega}\mid \mathcal{D})}{\pi_{\gamma}(\mathbf{\Omega}\mid \mathcal{D}^*)}
&=\frac{p(\mathcal{D}^*)}{p(\mathcal{D})}
\frac{\left(\sum_{i=1}^n \exp\left(-\frac{\gamma}{2}\bm{y}_i^{\top}\mathbf{\Omega} \bm{y}_i\right)\right)^{1/\gamma}}{\left(\sum_{i\in \mathcal{K}} \exp\left(-\frac{\gamma}{2}\bm{y}_i^{\top}\mathbf{\Omega} \bm{y}_i\right)\right)^{1/\gamma}},
\end{align*}
where 
\begin{align*}
p(\mathcal{D})&=\int  |\mathbf{\Omega}|^{1/2(1+\gamma)}\left[\sum_{i=1}^n \exp\left(-\frac{\gamma}{2}\bm{y}_i^{\top}\mathbf{\Omega} \bm{y}_i\right)\right]^{1/\gamma}\pi(\mathbf{\Omega}) d\mathbf{\Omega}, \\
p(\mathcal{D}^*)&=\int  |\mathbf{\Omega}|^{1/2(1+\gamma)}\left[\sum_{i\in \mathcal{K}} \exp\left(-\frac{\gamma}{2}\bm{y}_i^{\top}\mathbf{\Omega} \bm{y}_i\right)\right]^{1/\gamma}\pi(\mathbf{\Omega}) d\mathbf{\Omega}.
\end{align*}
We note that it holds that
\begin{align*}
\lim_{z\to\infty} \sum_{i=1}^n \exp\left(-\frac{\gamma}{2}\bm{y}_i^{\top}\mathbf{\Omega} \bm{y}_i\right)
&=\lim_{z\to\infty} \left[ \sum_{i\in\mathcal{K}} \exp\left(-\frac{\gamma}{2}\bm{y}_i^{\top}\mathbf{\Omega} \bm{y}_i\right)+\sum_{i\in\mathcal{L}} \exp\left(-\frac{\gamma}{2}\bm{y}_i^{\top}\mathbf{\Omega} \bm{y}_i\right)\right]\\
&=\sum_{i\in\mathcal{K}} \exp\left(-\frac{\gamma}{2}\bm{y}_i^{\top}\mathbf{\Omega} \bm{y}_i\right).
\end{align*}
Since the posterior distribution $\pi_{\gamma}(\mathbf{\Omega}\mid \mathcal{D})$ is proper, 
Lebesgue's dominated convergence theorem leads to the following convergence:
\begin{align*}
\lim_{z\to\infty} p(\mathcal{D})&=\lim_{z\to\infty}\int  |\mathbf{\Omega}|^{1/2(1+\gamma)}\left[\sum_{i=1}^n \exp\left(-\frac{\gamma}{2}\bm{y}_i^{\top}\mathbf{\Omega} \bm{y}_i\right)\right]^{1/\gamma}\pi(\mathbf{\Omega})d\mathbf{\Omega}\\
&=\int |\mathbf{\Omega}|^{1/(2(1+\gamma))}\left\{\lim_{z\to\infty}\left[\sum_{i=1}^n \exp\left(-\frac{\gamma}{2}\bm{y}_i^{\top}\mathbf{\Omega} \bm{y}_i\right)\right]^{1/\gamma}\right\}\pi(\mathbf{\Omega})d\mathbf{\Omega}\\
&=\int  |\mathbf{\Omega}|^{n/(2(1+\gamma))}\left[\sum_{i\in \mathcal{K}} \exp\left(-\frac{\gamma}{2}\bm{y}_i^{\top}\mathbf{\Omega} \bm{y}_i\right)\right]^{1/\gamma}\pi(\mathbf{\Omega})d\mathbf{\Omega}\\
&=p(\mathcal{D}^*).
\end{align*}
Then, we have 
\begin{align*}
\quad\lim_{z\to\infty} \frac{\pi_{\gamma}(\mathbf{\Omega}\mid \mathcal{D})}{\pi_{\gamma}(\mathbf{\Omega}\mid \mathcal{D}^*)}=1.
\end{align*}
Hence, it holds that
\begin{align*}
\lim_{z\to\infty}\int|\pi_{\gamma}(\mathbf{\Omega}\mid \mathcal{D})-\pi_{\gamma}(\mathbf{\Omega}\mid \mathcal{D}^*)|d\mathbf{\Omega}&=\int \lim_{z\to\infty}|\pi_{\gamma}(\mathbf{\Omega}\mid \mathcal{D})-\pi_{\gamma}(\mathbf{\Omega}\mid \mathcal{D}^*)|d\mathbf{\Omega}\\
&=\int \lim_{z\to\infty}\pi_{\gamma}(\mathbf{\Omega}\mid \mathcal{D}^*)\left|\frac{\pi_{\gamma}(\mathbf{\Omega}\mid \mathcal{D})}{\pi_{\gamma}(\mathbf{\Omega}\mid \mathcal{D}^*)}-1\right|d\mathbf{\Omega}\\
&=0.
\end{align*}
This completes the proof. 
\end{proof}
The result is interesting because the sufficient condition for posterior robustness is only the posterior propriety. In existing studies based on the (super) heavy-tailed probability distribution (e.g. \cite{desgagne2015robustness, Gagnon2020, hamura2022log}), the condition on the proportion of outliers is included in the sufficient conditions for the posterior robustness. The posterior robustness for other posterior distributions is discussed in Section \ref{sec:3}. 

We note that the posterior robustness under the proposed model leads to the robustness of point estimate in the frequentist $\gamma$-lasso method by \cite{hirose2017robust}, because the MAP estimate under the proposed method is theoretically equal to the $\gamma$-lasso estimate.

\subsection{Posterior computation}\label{subsec:2.4}

We provide an efficient posterior computation algorithm for the proposed posterior distribution \eqref{MAP-posterior}. We recall that the graphical lasso type prior is given by
\begin{align}
\pi(\mathbf{\Omega})&\propto\exp(-\lambda\|\mathbf{\Omega}\|_1)1_{\{\mathbf{\Omega}\in M^+\}}=\prod_{i=1}^p\mathrm{Exp}(w_{ii}\mid \lambda)\prod_{i<j}^p\mathrm{Lap}(w_{ij}\mid \lambda)1_{\{\mathbf{\Omega}\in M^+\}}.\label{Lasso-prior}
\end{align}
The MAP estimate of the proposed posterior distribution \eqref{MAP-posterior} under the prior \eqref{Lasso-prior} is equal to the estimate by \cite{hirose2017robust}. We note that the proposed model under the prior \eqref{Lasso-prior} satisfies the posterior propriety (Theorem \ref{posterior-proper}) and the posterior robustness (Theorem \ref{gamma-robust}). However, the posterior distribution is intractable because it involves the term $\left\{\sum_{i=1}^n\exp\left(-\frac{\gamma}{2}\bm{y}_i^{\top}\mathbf{\Omega} \bm{y}_i\right)\right\}^{1/\gamma}$ in the likelihood. Hence, we cannot construct an efficient Gibbs sampler in the proposed model. 

We employ an optimization-based sampling method called weighted Bayesian bootstrap (WBB) proposed by \cite{newton2021weighted}. The method gives an approximate posterior sample by adding a random perturbation to the MAP estimate. 
Following \cite{newton2021weighted}, we introduce a brief review of the WBB method. Let $\pi(\bm{\theta}\mid \bm{y})\propto \exp\{-(l(\bm{y}\mid \bm{\theta}) +\lambda \phi(\bm{\theta}))\}$ be a posterior distribution we want to compute, where $l(\bm{y}\mid \bm{\theta})=-\log f(\bm{y}\mid \bm{\theta})$ is a negative log-likelihood and $-\lambda \phi(\bm{\theta})=\log \pi(\bm{\theta})$ is a log-prior density with a hyperparameter $\lambda>0$. The WBB algorithm for sampling from the posterior distribution is conducted by iteratively optimizing the following randomized objective function. 
\begin{align*}
L_{\bm{w}}(\bm{\theta})=\left\{\sum_{i=1}^n w_i l(y_i\mid \bm{\theta})\right\} + w_0 \lambda \phi(\bm{\theta}),
\end{align*}
where $\bm{w}=(w_0,w_1,\dots,w_n)^\top$ is a random weight vector from the Dirichlet distribution $(n+1)\mathrm{Dirichlet}(1,\dots,1)$. Explicit pseudocode for the WBB algorithm is provided in Algorithm 1 of \cite{newton2021weighted}, and some implementation examples including the Bayesian lasso and Bayesian trend filtering are also given in \cite{newton2021weighted}.
After $M$ iterations, we can obtain the approximate posterior sample with sample size $M$ as $\{\bm{\theta}^{(1)},\dots,\bm{\theta}^{(M)}\}$. Since the sampling algorithm is not based on the Markov chain, it is approximation-based sampling method. The randomized objective function for the proposed $\gamma$-posterior under the prior \eqref{Lasso-prior} is defined by 
\begin{align}\label{BB-objective}
L_{\bm{w}}(\mathbf{\Omega})&=-\frac{1}{\gamma}\log\left\{\frac{1}{n}\sum_{i=1}^nw_if(\bm{y}_i\mid \mathbf{\Omega})^{\gamma}\right\}+\frac{\gamma}{2(1+\gamma)}\log|\mathbf{\Omega}|+w_0\lambda\|\mathbf{\Omega}\|_1,
\end{align}
where $\bm{w}=(w_0,w_1,\dots,w_n)^\top \sim (n+1)\mathrm{Dirichlet}(1,\dots,1)$, and $f(\bm{y}_i\mid \mathbf{\Omega})$ is the density function of the multivariate Gaussian distribution. Large sample asymptotic property of the WBB posterior is discussed by \cite{newton2021weighted}, and they show that the posterior distribution generated from the WBB algorithm has asymptotic normality.
The algorithm was also studied in \cite{nie2022bayesian}, and they showed a posterior concentration result for the WBB posterior under the normal linear regression models. Although a similar theoretical result could be expected to hold for each component of the inverse of the covariance matrix of the multivariate Gaussian distribution, this is beyond the scope of this paper and is a subject for future work. In particular, the asymptotic property of the WBB posterior distribution for large graphs with $p$ growing together with $n$ is an interesting future problem.
To solve the minimization problem of \eqref{BB-objective}, we employed the Majorize-Minimization (MM) algorithm (see also \cite{hirose2017robust}). By using Jensen's inequality, we can show that the weighted objective function \eqref{BB-objective} is evaluated by 
\begin{align}\label{bound}
L_{\bm{w}}(\mathbf{\Omega})\le L_{\bm{w}}^*(\mathbf{\Omega})\propto \mathrm{tr}\left\{\mathbf{S}^*\mathbf{\Omega}\right\}-\log|\mathbf{\Omega}|+\rho\|\mathbf{\Omega}\|_1
\end{align}
where
\begin{align*}
\mathbf{S}^*&=(1+\gamma)\sum_{i=1}^ns_i^*\bm{y}_i\bm{y}_i^{\top},\quad s_i^*=\frac{w_if(\bm{y}_i\mid\mathbf{\Omega})^{\gamma}}{\sum_{j=1}^nw_jf(\bm{y}_j\mid\mathbf{\Omega})^{\gamma}},\quad \rho=2(1+\gamma)\lambda w_0.
\end{align*}
The derivation of \eqref{bound} is given in \ref{subsec:B1}. 
Under some regularity conditions, asymptotic properties of the approximate posterior distribution via the WBB were shown for a sufficiently large $n$ (see e.g., \cite{lyddon2019general, newton2021weighted}). 
To optimize the right-hand side of \eqref{bound}, we used an excellent algorithm proposed by \cite{friedman2008sparse}. The proposed weighted Bayesian bootstrap algorithm is summarized in Algorithm \ref{algo:MMWBB}. An important point of the algorithm is capable of parallel computation which is different from Markov chain Monte Carlo (MCMC) methods. The penalty parameter $\lambda$ is fixed as \cite{hirose2017robust} because the selection of $\lambda$ in the presence of outliers is not straightforward. As a remedy, we may be able to sample $\lambda$ from the posterior distribution in the similar way to \cite{hashimoto2020robust}. We attempt to estimate $\lambda$ from the data through weighted Bayesian bootstrap within Gibbs sampler in the Supplementary Materials.

\begin{algorithm*}[thbp] 
\caption{\bf--- Weighted Bayesian bootstrap via MM algorithm.}
\label{algo:MMWBB}
Set a tuning parameter $\lambda$, a threshold $\varepsilon'$ for convergence,  and an initial value $\mathbf{\Omega}^{(0)}$ of the MM algorithm (e.g. \cite{hirose2017robust}).
\begin{itemize}
\item[1] Generate a random vector $\bm{w}=(w_0,w_1,\dots,w_n)^\top\sim (n+1)\mathrm{Dirichlet}(1,\dots,1)$.
\item[2] Optimize the weighted objective function \eqref{BB-objective} using the following MM algorithm.
\begin{itemize}
\item[(i)] Compute $s_i^{*(t)}$ and $\mathbf{S}^{*(t)}$ based on $\mathbf{\Omega}^{(t)}$, and the calculate $\mathbf{\Omega}^{(t+1)}$ by optimizing the function \eqref{bound} (e.g. \cite{friedman2008sparse}).
\item[(ii)] If $\max_{i,j} |\omega_{ij}^{(t+1)}-\omega_{ij}^{(t)}|<\varepsilon'$, then stop the algorithm and set $\mathbf{\Omega}^{(t)}$ as the optimal value $\hat{\mathbf{\Omega}}$. If $\max_{i,j} |\omega_{ij}^{(t+1)}-\omega_{ij}^{(t)}|>\varepsilon'$, go to the next steps.
\item[(iii)] Set $t$ with $t+1$, and go back to step (i).
\end{itemize}
\item[3] Cycle steps 1 and 2, and get the $m$th posterior sample.
\end{itemize}
\end{algorithm*}

\subsection{Bayesian inference on graphical structures}
\label{subsec:2.5}

Since the original graphical lasso provides a sparse solution $\hat{\omega}_{ij}=0$ for $i\ne j$ by solving the optimization problem, we can directly estimate the dependence structure of the Gaussian graphical models. The Bayesian graphical lasso \citep{wang2012bayesian} cannot produce such a sparse solution because the posterior probability of the event $\{\omega_{ij}=0\}$ is zero due to the use of the continuous shrinkage priors. Spike-and-slab type priors might be useful to obtain an exact zero solution, but it is known that there are some computational issues. As a remedy, \cite{wang2012bayesian} proposed a variable selection method based on the posterior mean estimator of the partial correlation via the amount of shrinkage. Although the method seems to work reasonably well, the method needs to assume the prior distribution for non-zero $\omega_{ij}$, and \cite{wang2012bayesian} assumed the standard conjugate Wishart prior $W(3, \mathbf{I}_p)$.

In our algorithm (Algorithm \ref{algo:MMWBB}), we can obtain an element-wise sparse solution for $\mathbf{\Omega}=(\omega_{ij})$ in each iteration of the WBB algorithm, which is a mimic of the spike-and-slab strategy. Therefore, we consider a variable selection method via the proportion of non-zero $\omega_{ij}$ in the posterior sample (i.e. posterior probability of $\{\omega_{ij}\ne 0\}$). In our numerical experiments, however, we search a significant dependence via the criterion so-called median probability criterion (see e.g., \cite{berger2004optimal}) defined by
\begin{align*}
\mathrm{P}(|\omega_{ij}| < \varepsilon \mid \mathbf{Y}) \ge 0.5,
\end{align*}
where $\varepsilon>0$ is a threshold. In numerical experiments, the threshold $\varepsilon$ is set as $10^{-2}$ as a default choice, which is also adopted in the frequentist method (e.g. \cite{fan2009network}). 
In practice, the posterior probability is approximated using the posterior sample. 
Uncertainty quantification based on the method is illustrated in Section \ref{sec:5}.

\section{Comparison with existing methods}\label{sec:3}

We discuss three Bayesian Gaussian graphical models in terms of posterior robustness. The proofs of the propositions are given in \ref{sec:A}.

\subsection{Bayesian graphical lasso}\label{subsec:3.1}

The standard posterior distribution of $\mathbf{\Omega}$ under the Gaussian likelihood is regarded as one based on the Kullback--Leibler (KL) divergence. The posterior is defined by
\begin{align}
\pi_{\mathrm{KL}}(\mathbf{\Omega}\mid \mathbf{Y})=\frac{|\mathbf{\Omega}|^{n/2}\exp(-\sum_{i=1}^n \bm{y}_i^{\top}\mathbf{\Omega} \bm{y}_i/2)\pi(\mathbf{\Omega})}{\int |\mathbf{\Omega}|^{n/2}\exp(-\sum_{i=1}^n \bm{y}_i^{\top}\mathbf{\Omega} \bm{y}_i/2)\pi(\mathbf{\Omega})d\mathbf{\Omega}}.\label{KL-pos}
\end{align}
\cite{wang2012bayesian} proposed a Bayesian graphical lasso whose posterior distribution is defined by \eqref{KL-pos} with the Laplace prior \eqref{Lasso-prior}, and also provided an efficient block Gibbs sampler to calculate the corresponding posterior distribution. However, the posterior distribution does not satisfy the posterior robustness due to the Gaussian assumption. 

\begin{proposition}\label{KL-robust}
Assume that the standard posterior $\pi_{\mathrm{KL}}(\mathbf{\Omega}\mid \mathbf{Y})$ in \eqref{KL-pos} is proper for all observations. Then, the posterior robustness does not hold.
\end{proposition}

\subsection{Bayesian $t$-graphical lasso}\label{subsec:3.2}

As we mentioned in Section \ref{sec:2}, heavy-tailed probability distributions are often used in classical robust statistics. 
We employ the multivariate $t$-distribution instead of the multivariate Gaussian distribution. \cite{finegold2011robust} proposed the graphical model based on the multivariate $t$-distribution. The corresponding posterior distribution is defined by
\begin{align}
\pi_{t}(\mathbf{\Omega}\mid \mathbf{Y})=\frac{ \prod_{i=1}^n|\mathbf{\Omega}|^{1/2}\left(1+\frac{1}{\nu-2}\bm{y}_i^{\top}\mathbf{\Omega} \bm{y}_i\right)^{-(\nu+p)}\pi(\mathbf{\Omega})}{\int \prod_{i=1}^n|\mathbf{\Omega}|^{1/2}\left(1+\frac{1}{\nu-2}\bm{y}_i^{\top}\mathbf{\Omega} \bm{y}_i\right)^{-(\nu+p)} \pi(\mathbf{\Omega})d\mathbf{\Omega}},\label{t-pos}
\end{align}
where $\nu>2$ is a degree of freedom and $\mathbf{\Omega}$ is a precision matrix. Assuming the graphical lasso prior in \eqref{t-pos}, we obtain the Bayesian $t$-graphical lasso model. Since the multivariate $t$-distribution can be represented as scale mixtures of normal distribution, we can construct an block Gibbs sampler in a similar way to \cite{wang2012bayesian}. Unfortunately, we have the following result on the posterior distribution \eqref{t-pos}. 

\begin{proposition}\label{t-robust}
Assume that the posterior $\pi_{t}(\mathbf{\Omega}\mid \mathbf{Y})$ in \eqref{t-pos} is proper for all observations. Then, the posterior robustness does not hold.
\end{proposition}
The result is consistent with \cite{desgagne2015robustness}. They showed that the $t$-distribution does not lead to posterior robustness for joint estimation of location and scale parameters. Furthermore, the variance of the estimator based on $t$-graphical lasso tends to be large because of the heaviness of the tail (see also \cite{hirose2017robust}).

\subsection{Bayesian graphical lasso via density-power divergence}\label{subsec:3.3}

Although our main proposal is based on the $\gamma$-divergence, the density-power divergence is also known as the robust divergence \citep{basu1998robust}. \cite{sun2012robust} proposed a robust graphical lasso based on the density-power divergence, while \cite{ghosh2016robust} proposed a robustified posterior distribution based on the density-power divergence. Since \cite{ghosh2016robust} dealt with only univariate observations, the Bayesian formulation of graphical lasso via the density-power divergence has not been considered. In Gaussian graphical model, the density-power posterior $\pi_{\mathrm{DP}}(\mathbf{\Omega}\mid \mathbf{Y})$ under a prior $\pi(\mathbf{\Omega})$ is defined by
\begin{align}
\pi_{\mathrm{DP}}(\mathbf{\Omega}\mid \mathbf{Y})=\frac{\exp\left(Q_n^{(\alpha)}(\mathbf{\Omega})\right)\pi(\mathbf{\Omega})}{\int \exp\left(Q_n^{(\alpha)}(\mathbf{\Omega})\right)\pi(\mathbf{\Omega})d\mathbf{\Omega}},\label{DP-pos}
\end{align}
where 
\begin{align*}
Q_n^{(\alpha)}(\mathbf{\Omega})&=\frac{1}{\alpha}\sum_{i=1}^nf(\bm{y}_i\mid \mathbf{\Omega})^{\alpha}-\frac{n}{1+\alpha}\int f(\bm{y}\mid\mathbf{\Omega})^{1+\alpha}d\bm{y}\notag\\
&=(2\pi)^{-\alpha p/2}|\mathbf{\Omega}|^{\alpha/2}\left[\frac{1}{\alpha}\sum_{i=1}^n\exp\left(-\frac{\alpha}{2}\bm{y}_i^{\top}\mathbf{\Omega} \bm{y}_i\right)-n(1+\alpha)^{1-\alpha/2}\right],
\end{align*}
and $\alpha>0$ is a tuning parameter which plays the same role as $\gamma>0$ in the $\gamma$-divergence. Note that the density-power posterior also converges to the standard posterior \eqref{KL-pos} as $\alpha\to0$.

\begin{proposition}\label{DP-robust}
Assume that the posterior $\pi_{\mathrm{DP}}(\mathbf{\Omega}\mid \mathbf{Y})$ in \eqref{DP-pos} is proper for all observations. Then, the posterior robustness does not hold.
\end{proposition}
In general, it is known that the minimum density-power divergence estimate does not work well for estimating the scale (or variance) parameter in the univariate case (see, e.g., \cite{fujisawa2008robust, nakagawa2020robust}). Hence, the result in Proposition \ref{DP-robust} is consistent with such previous observations.

\section{Simulation studies}\label{sec:4}

In this section, we illustrate the performance of the proposed model through some numerical experiments. 

\subsection{Simulation setting}\label{subsec:4.1}
The following three data-generating processes were considered: 
\begin{itemize}
\item[(a)] $\mathcal{N}_p(\bm{0},\mathbf{\Omega}^{-1})$,
\item[(b)] $(1-\varepsilon)\mathcal{N}_p(\bm{0},\mathbf{\Omega}^{-1})+\varepsilon \mathcal{N}_p(\bm{0}, 30\mathbf{I}_p)$,
\item[(c)] $(1-\varepsilon)\mathcal{N}_p(\bm{0},\mathbf{\Omega}^{-1})+\varepsilon \mathcal{N}_p(\eta \bm{1}^{(3)}, \mathbf{I}_p)$,
\end{itemize}
where $\varepsilon \in [0,1)$ is the contamination ratio, $\mathbf{I}_p$ is the identity matrix, and $\bm{1}^{(3)}=(1,1,1,0,\dots,0)^{\top}$ is a $p$-dimensional vector whose first 3 elements are 0 and the other $p-3$ elements are 0. Note that the model (a) has no outliers in the dataset. Models (b) and (c) generate outliers from distributions with symmetric large variance and partially large mean. Such scenarios were also considered in \cite{hirose2017robust}. For model (c), the larger the $\eta$, the larger the outliers. We set $\eta=5,10,20$, $\varepsilon=0.1$, and $p=12$. 
We considered the following true sparse precision matrices:
\begin{itemize}
    \item[(A)] The precision matrix dealt with in Subsection 5.2 in \cite{wang2015scaling} (for details, see \ref{subsec:C1}).
    \item[(B)] The AR(2) structure,  where $\omega_{ii}=1$, $\omega_{i,i-1}=\omega_{i-1,1}=0.5$, and $\omega_{i,i-2}=\omega_{i-2,i}=0.25$ for $i=1,\dots, p$ (see also \cite{wang2012bayesian}).
\end{itemize}
For example, if the data-generating process is (a) and the true precision matrix is (A), we use the notation like (a-A) for simplicity. 
For each scenario, the sample size was set as $n=200$. We compared the following methods:
\begin{itemize}
\item BR: The proposed robust Bayesian graphical lasso via weighted Bayesian bootstrap. We generated 6000 posterior samples via the WBB algorithm.
\item BT: Robust Bayesian graphical lasso under multivariate $t$-distribution in which the degree of freedom is 3. The Gibbs sampler can be directly derived by using \cite{wang2012bayesian}'s algorithm and it is given in \ref{subsec:B2}. We generated 6000 posterior samples after discarding the first 4000 samples as burn-in.
\item BG: (Non-robust) Bayesian graphical lasso proposed by \citep{wang2012bayesian}, which is based on Gaussian likelihood. The Gibbs sampler is implemented using {\tt R} package {\tt BayesianGLasso}. We generated 6000 posterior samples after discarding the first 4000 samples as burn-in.
\item FR: Frequentist robust graphical lasso based on the $\gamma$-divergence by \cite{hirose2017robust}. The MM algorithm is applied to this method as the BR method and the method can also be implemented using their {\tt R} package {\tt rsggm}.
\item FG: Frequentist (non-robust) graphical lasso proposed by \cite{friedman2008sparse}, which is based on Gaussian likelihood. The method can be implemented using the {\tt R} package {\tt glasso} or {\tt glassoFast}.
\end{itemize}
We consider five values on tuning parameter as $\lambda \in \{\lambda_{\mathrm{min}}\times i\mid i=1,\dots,5 \}$ in the BR, FR, and FG methods, where $\lambda_{\mathrm{min}}=0.02$ for the BR and FR methods, and $\lambda_{\mathrm{min}}=0.04$ for the FG method. The BR, FR, and FG methods for five tuning parameters are denoted by BR$i$, FR$i$, FG$i$ ($i=1,\dots,5$) for short. Note that we can estimate $\lambda$ through the Gibbs sampler in the BT and BG methods. 
For BR and FR, we set $\gamma=0.05$ and $0.1$. To evaluate the performance, we calculate the square root of the mean squared error (RMSE), the average length of 95\% credible interval (AL), and the coverage probability of 95\% credible interval (CP) under 100 repetitions:
\begin{align*}
\mathrm{RMSE}&=\left\{\frac{1}{(p-1)p/2}\sum_{j=1}^p\sum_{i>j} (\omega_{ij}-\hat{\omega}_{ij})^2 \right\}^{1/2},\\ 
\mathrm{AL}&=\frac{1}{(p-1)p/2}\sum_{j=1}^p\sum_{i>j} \hat{\omega}_{ij}^{(97.5)}-\hat{\omega}_{ij}^{(2.5)},\\
\mathrm{CP}&=\frac{1}{(p-1)p/2}\sum_{j=1}^p\sum_{i>j} I(\hat{\omega}_{ij}^{(2.5)}\leq \omega_{ij} \leq \hat{\omega}_{ij}^{(97.5)}),
\end{align*}
where $\hat{\omega}_{ij}$ is the point estimate of $\omega_{ij}$, and $\omega_{ij}^{(\alpha)}$ is the $100\alpha$\% posterior quantiles of $\omega_{ij}$. Note that the AL and CP are only reported for Bayesian methods, and the point estimates for Bayesian methods are the mean of posterior samples. For the BR, FR and FG methods, we calculate true and false positive rates (TPR/FPR) and false discovery rate (FDR) for each tuning parameter $\lambda$. The criterion for determining whether an element is 0 or not for the BR method is given in Section \ref{subsec:2.5}.

\subsection{Simulation results}\label{subsec:4.2}
We show examples of simulated posterior distributions of $\omega_{2,5}$ in Figure \ref{oneshot-posterior-simulation}. These posterior distributions are based on the data generated from the scenario (c-B). It is observed that the proposed $\gamma$-posterior distributions for $\gamma=0.05$ and $0.1$ hold the posterior robustness under large outliers such as $\eta=20$, whlie the BG posterior distribution in the right panel is affected by outliers. 
The results of Monte Carlo simulations for all scenarios are summarized in Figure~\ref{sim-RMSE}, Tables~\ref{AL-CP} and \ref{sim-TPR-FPR}. Note that we summarize only the case of $\gamma=0.1$ and the rest of the results are summarized in \ref{subsec:C2}. From the RMSE for scenario (a), almost all methods except for the BT method work well, although the non-robust BG and FG methods have the smallest values. The robust methods performed the best in terms of point estimates in most cases. In particular, scenario (b) is the scenario in which there is a marked difference between robust and non-robust methods. For all scenarios, the BT method is the worst because the estimated matrix does not correspond to the Gaussian precision matrix. For scenario (c), the boxplots and the median lines for the BR and FR methods move down as $\eta$ increases due to robustness under extremely large outliers. In terms of uncertainty quantification, the proposed BR methods have stable CP values for most cases. In particular, as seen in (b-B), the BG method gives the worst values less than 50\% while the BR method gives the same values as the other scenario. For the BR method, the AL and CP are smaller as the tuning parameter is large. As the magnitude of outliers is higher, the CP values get smaller because of the posterior robustness. The TPR, FPR, and FDR are reported in Table~\ref{sim-TPR-FPR}. The TPR and FPR for all methods decreased as $\lambda$ increased due to stronger shrinkage. It seems that the proposed criterion for shrinkage works well for all scenarios. In particular, the matrix (B) gives clear differences between the robust and non-robust methods as the results of RMSE. Compared to the BR and FG methods, the FG method performed much worse, especially in scenario (b). We provide additional simulations including different dimensions and scenarios, such as small-world graphs and scale-free networks, in the Supplementary Materials.

\begin{figure}[t]
\begin{center}
\includegraphics[width=\linewidth]{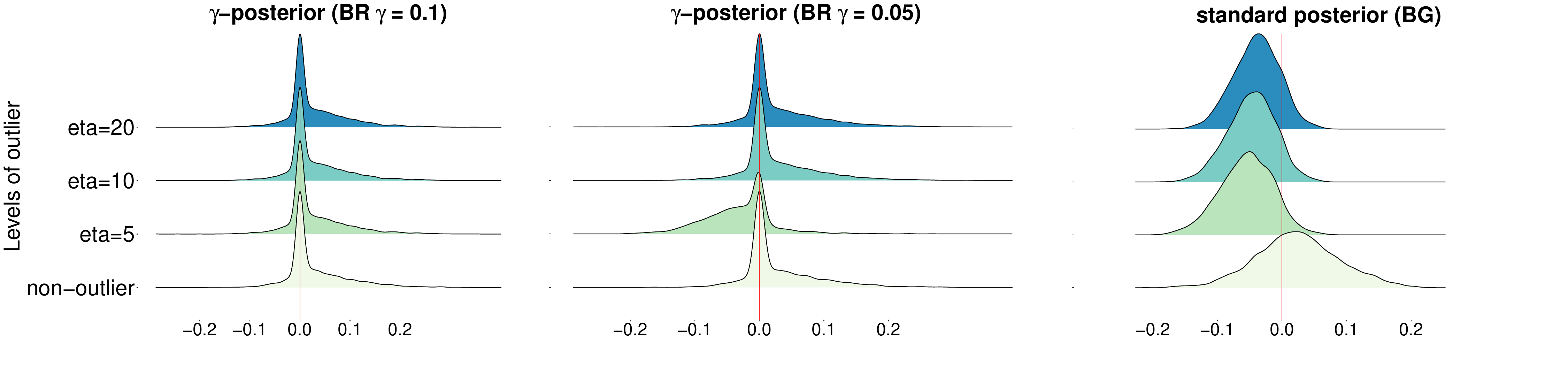}
\caption{One-shot simulation result for the three posterior distributions of $\omega_{2,5}$ in the scenario (c-B).}
\label{oneshot-posterior-simulation}
\end{center}
\end{figure}

\begin{figure}[htbp]
\begin{center}
\includegraphics[width=\linewidth]{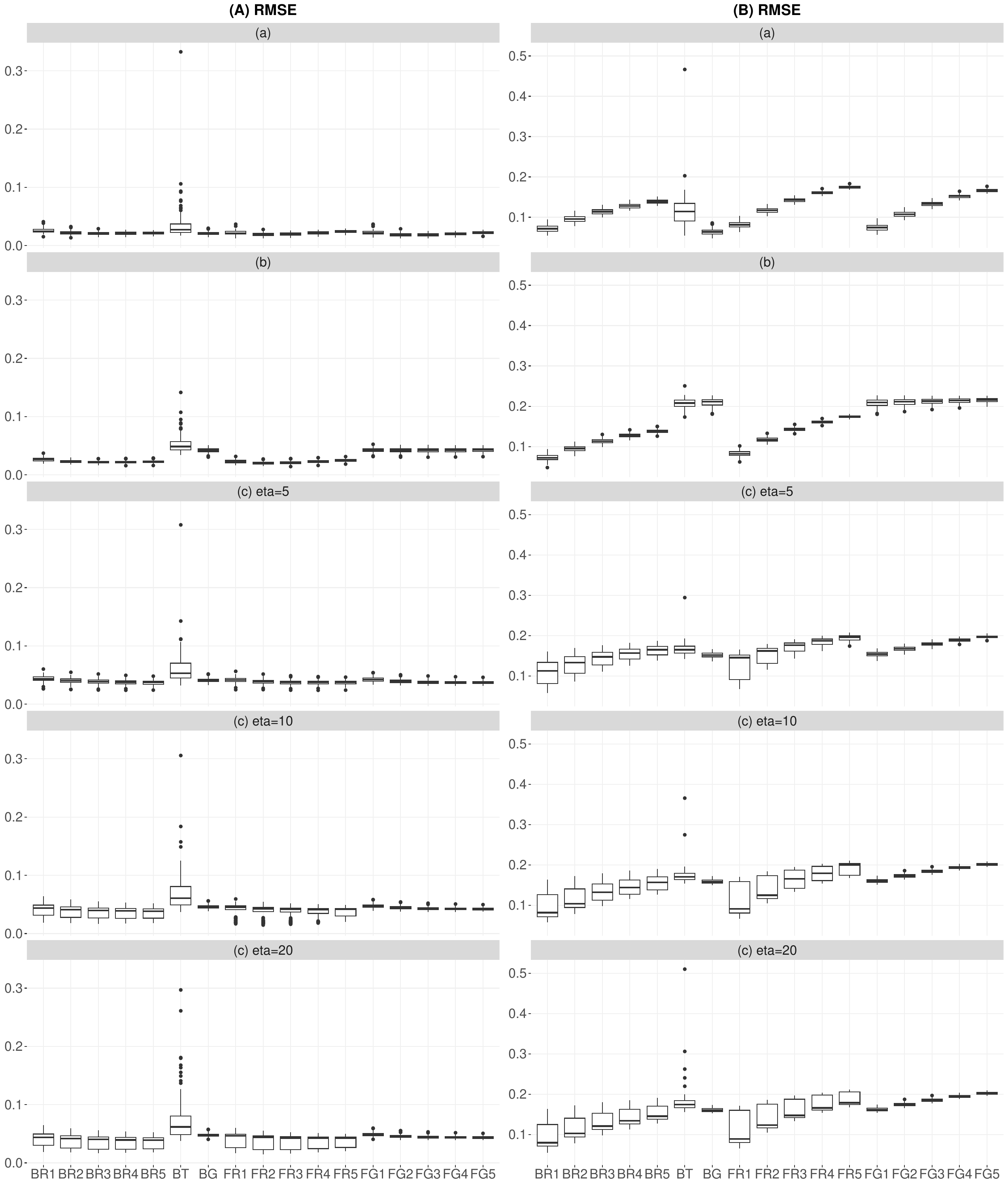}
\caption{The boxplots of RMSE based on 100 repetitions for all scenarios are summarized. The left and right panels correspond to the results for (A) and (B), respectively, and the results for the five data-generating processes are summarized from top to bottom.}
\label{sim-RMSE}
\end{center}
\end{figure}

\begin{table}[htbp]
\caption{Coverage probabilities and average lengths of 95\% credible intervals averaged over 100 Monte Carlo replications.}
\begin{center}
\begin{tabular}{cc|ccccccc}
  \toprule
  \multicolumn{9}{c}{Data-generating process (a)} \\
  \hline
  &&BR1 & BR2 & BR3 & BR4 & BR5 &BT&BG\\
  \hline
\multirow{2}{*}{(A)}& AL & 0.106 & 0.097 & 0.091 & 0.086 & 0.082 & 0.209 & 0.087 \\ 
   & CP & 0.929 & 0.911 & 0.888 & 0.865 & 0.842 & 0.980 & 0.968 \\ 
\multirow{2}{*}{(B)} & AL & 0.281 & 0.265 & 0.251 & 0.238 & 0.226 & 0.358 & 0.241 \\ 
 & CP & 0.906 & 0.853 & 0.803 & 0.758 & 0.712 & 0.910 & 0.942 \\  
\midrule
\multicolumn{9}{c}{Data-generating process (b)} \\
   \hline
  &&BR1 & BR2 & BR3 & BR4 & BR5 &BT&BG\\
  \hline
\multirow{2}{*}{(A)}& AL & 0.112 & 0.102 & 0.095 & 0.090 & 0.086 & 0.192 & 0.046 \\ 
   & CP & 0.938 & 0.916 & 0.897 & 0.878 & 0.862 & 0.830 & 0.734 \\ 
  \multirow{2}{*}{(B)}& AL & 0.297 & 0.277 & 0.262 & 0.248 & 0.236 & 0.212 & 0.078 \\ 
 & CP & 0.916 & 0.867 & 0.821 & 0.775 & 0.734 & 0.517 & 0.473 \\ 
  \midrule
  \multicolumn{9}{c}{Data-generating process (c) $\eta=5$} \\
   \hline
  &&BR1 & BR2 & BR3 & BR4 & BR5 &BT&BG\\
  \hline
\multirow{2}{*}{(A)}& AL & 0.099 & 0.091 & 0.086 & 0.082 & 0.079 & 0.210 & 0.079 \\ 
   & CP & 0.901 & 0.884 & 0.864 & 0.844 & 0.828 & 0.943 & 0.928 \\ 
  \multirow{2}{*}{(B)} & AL & 0.286 & 0.264 & 0.247 & 0.232 & 0.219 & 0.307 & 0.193 \\ 
   & CP & 0.892 & 0.833 & 0.779 & 0.729 & 0.678 & 0.813 & 0.772 \\  
  \midrule
  \multicolumn{9}{c}{Data-generating process (c) $\eta=10$} \\
    \hline
  &&BR1 & BR2 & BR3 & BR4 & BR5 &BT&BG\\
  \hline
\multirow{2}{*}{(A)}& AL & 0.107 & 0.098 & 0.092 & 0.087 & 0.084 & 0.207 & 0.076 \\ 
   & CP & 0.928 & 0.911 & 0.888 & 0.873 & 0.854 & 0.940 & 0.926 \\ 
  \multirow{2}{*}{(B)} & AL & 0.289 & 0.268 & 0.252 & 0.238 & 0.226 & 0.309 & 0.189 \\ 
   & CP & 0.905 & 0.848 & 0.800 & 0.749 & 0.701 & 0.809 & 0.772 \\ 
\midrule
  \multicolumn{9}{c}{Data-generating process (c) $\eta=20$} \\
    \hline
  &&BR1 & BR2 & BR3 & BR4 & BR5 &BT&BG\\
  \hline
\multirow{2}{*}{(A)}& AL & 0.107 & 0.098 & 0.092 & 0.087 & 0.084 & 0.199 & 0.075 \\ 
   & CP & 0.926 & 0.910 & 0.889 & 0.873 & 0.855 & 0.935 & 0.925 \\ 
  \multirow{2}{*}{(B)} & AL & 0.290 & 0.270 & 0.253 & 0.239 & 0.226 & 0.305 & 0.188 \\ 
   & CP & 0.906 & 0.853 & 0.804 & 0.756 & 0.705 & 0.788 & 0.771 \\  
\bottomrule
\end{tabular}
\end{center}

\label{AL-CP}
\end{table}

\begin{table}[htbp]

\caption{True positive rates, false positive rates, and false discovery rates averaged over 100 Monte Carlo replications.}
\begin{center}
\resizebox{1.0\textwidth}{!}{ 
\begin{tabular}{cc|ccccc|ccccc|ccccc|c}
  \toprule
  \multicolumn{17}{c}{Data-generating process (a)} \\
  \hline
  && BR1 & BR2 & BR3 & BR4 & BR5 & FR1 & FR2 & FR3 & FR4 & FR5 & FG1 & FG2 & FG3 & FG4 & FG5 & BG \\
  \hline
\multirow{3}{*}{(A)}& TPR & 0.91 & 0.82 & 0.77 & 0.73 & 0.70 & 0.96 & 0.93 & 0.90 & 0.86 & 0.83 & 0.97 & 0.94 & 0.91 & 0.89 & 0.86 & 0.43 \\ 
   & FPR & 0.89 & 0.46 & 0.29 & 0.20 & 0.14 & 0.80 & 0.64 & 0.51 & 0.41 & 0.33 & 0.83 & 0.68 & 0.55 & 0.46 & 0.39 & 0.28 \\ 
   & FDR & 0.02 & 0.04 & 0.06 & 0.07 & 0.07 & 0.01 & 0.02 & 0.03 & 0.03 & 0.04 & 0.01 & 0.02 & 0.02 & 0.03 & 0.03 & 0.14 \\ 
  \multirow{3}{*}{(B)} & TPR & 1.00 & 1.00 & 0.98 & 0.95 & 0.91 & 1.00 & 0.99 & 0.97 & 0.92 & 0.86 & 1.00 & 1.00 & 0.98 & 0.95 & 0.91 & 1.00 \\ 
   & FPR & 0.97 & 0.40 & 0.26 & 0.19 & 0.15 & 0.61 & 0.43 & 0.32 & 0.25 & 0.20 & 0.65 & 0.47 & 0.36 & 0.28 & 0.24 & 0.36 \\ 
   & FDR & 0.00 & 0.00 & 0.01 & 0.02 & 0.04 & 0.00 & 0.00 & 0.01 & 0.04 & 0.07 & 0.00 & 0.00 & 0.01 & 0.02 & 0.04 & 0.00 \\ 
\midrule
  \multicolumn{17}{c}{Data-generating process (b)} \\
  \hline
  && BR1 & BR2 & BR3 & BR4 & BR5 & FR1 & FR2 & FR3 & FR4 & FR5 & FG1 & FG2 & FG3 & FG4 & FG5 & BG\\
  \hline
\multirow{3}{*}{(A)}& TPR & 0.92 & 0.83 & 0.78 & 0.73 & 0.70 & 0.96 & 0.93 & 0.90 & 0.87 & 0.82 & 0.97 & 0.95 & 0.92 & 0.90 & 0.87 & 0.38 \\ 
   & FPR & 0.92 & 0.52 & 0.32 & 0.21 & 0.16 & 0.80 & 0.65 & 0.53 & 0.42 & 0.34 & 0.95 & 0.90 & 0.85 & 0.81 & 0.76 & 0.24 \\ 
   & FDR & 0.02 & 0.04 & 0.05 & 0.07 & 0.07 & 0.01 & 0.02 & 0.03 & 0.03 & 0.04 & 0.01 & 0.01 & 0.02 & 0.02 & 0.03 & 0.15 \\ 
  \multirow{3}{*}{(B)} & TPR & 1.00 & 0.99 & 0.98 & 0.94 & 0.90 & 1.00 & 0.99 & 0.97 & 0.92 & 0.85 & 0.96 & 0.91 & 0.87 & 0.83 & 0.79 & 0.66 \\ 
   & FPR & 0.99 & 0.44 & 0.28 & 0.2 & 0.15 & 0.63 & 0.44 & 0.33 & 0.25 & 0.20 & 0.94 & 0.87 & 0.82 & 0.77 & 0.72 & 0.35 \\ 
   & FDR & 0.00 & 0.00 & 0.01 & 0.03 & 0.05 & 0.00 & 0.00 & 0.02 & 0.04 & 0.07 & 0.02 & 0.04 & 0.06 & 0.08 & 0.10 & 0.16 \\ 
  \midrule
    \multicolumn{17}{c}{Data-generating process (c) $\eta=5$} \\
  \hline
  && BR1 & BR2 & BR3 & BR4 & BR5 & FR1 & FR2 & FR3 & FR4 & FR5 & FG1 & FG2 & FG3 & FG4 & FG5 & BG\\
  \hline
\multirow{3}{*}{(A)}& TPR & 0.92 & 0.83 & 0.78 & 0.73 & 0.70 & 0.96 & 0.93 & 0.90 & 0.87 & 0.82 & 0.97 & 0.95 & 0.92 & 0.90 & 0.87 & 0.38 \\ 
   & FPR & 0.92 & 0.52 & 0.32 & 0.21 & 0.16 & 0.80 & 0.65 & 0.53 & 0.42 & 0.34 & 0.95 & 0.90 & 0.85 & 0.81 & 0.76 & 0.24 \\ 
   & FDR & 0.02 & 0.04 & 0.05 & 0.07 & 0.07 & 0.01 & 0.02 & 0.03 & 0.03 & 0.04 & 0.01 & 0.01 & 0.02 & 0.02 & 0.03 & 0.15 \\ 
  \multirow{3}{*}{(B)} & TPR & 1.00 & 0.99 & 0.98 & 0.94 & 0.90 & 0.01 & 0.99 & 0.97 & 0.92 & 0.85 & 0.96 & 0.91 & 0.87 & 0.83 & 0.79 & 0.66 \\ 
   & FPR & 0.99 & 0.44 & 0.28 & 0.20 & 0.15 & 0.63 & 0.44 & 0.33 & 0.25 & 0.20 & 0.94 & 0.87 & 0.82 & 0.77 & 0.72 & 0.35 \\ 
   & FDR & 0.00 & 0.00 & 0.01 & 0.03 & 0.05 & 0.00 & 0.00 & 0.02 & 0.04 & 0.07 & 0.02 & 0.04 & 0.06 & 0.08 & 0.10 & 0.16 \\ 
  \midrule
    \multicolumn{17}{c}{Data-generating process (c) $\eta=10$} \\
  \hline
  && BR1 & BR2 & BR3 & BR4 & BR5 & FR1 & FR2 & FR3 & FR4 & FR5 & FG1 & FG2 & FG3 & FG4 & FG5 & BG\\
  \hline
\multirow{3}{*}{(A)} & TPR & 0.92 & 0.83 & 0.78 & 0.73 & 0.70 & 0.96 & 0.93 & 0.90 & 0.87 & 0.82 & 0.97 & 0.95 & 0.92 & 0.90 & 0.87 & 0.38 \\ 
   & FPR & 0.92 & 0.52 & 0.32 & 0.21 & 0.16 & 0.8 & 0.65 & 0.53 & 0.42 & 0.34 & 0.95 & 0.90 & 0.85 & 0.81 & 0.76 & 0.24 \\ 
   & FDR & 0.02 & 0.04 & 0.05 & 0.07 & 0.07 & 0.01 & 0.02 & 0.03 & 0.03 & 0.04 & 0.01 & 0.01 & 0.02 & 0.02 & 0.03 & 0.15 \\ 
  \multirow{3}{*}{(B)} & TPR & 1.00 & 0.99 & 0.98 & 0.94 & 0.90 & 1.00 & 0.99 & 0.97 & 0.92 & 0.85 & 0.96 & 0.91 & 0.87 & 0.83 & 0.79 & 0.66 \\ 
   & FPR & 0.99 & 0.44 & 0.28 & 0.20 & 0.15 & 0.63 & 0.44 & 0.33 & 0.25 & 0.20 & 0.94 & 0.87 & 0.82 & 0.77 & 0.72 & 0.35 \\ 
   & FDR & 0.00 & 0.00 & 0.01 & 0.03 & 0.05 & 0.00 & 0.00 & 0.02 & 0.04 & 0.07 & 0.02 & 0.04 & 0.06 & 0.08 & 0.10 & 0.16 \\ 
\midrule
\multicolumn{17}{c}{Data-generating process (c) $\eta=20$} \\
\hline
  && BR1 & BR2 & BR3 & BR4 & BR5 & FR1 & FR2 & FR3 & FR4 & FR5 & FG1 & FG2 & FG3 & FG4 & FG5 & BG\\
  \hline
\multirow{3}{*}{(A)}& TPR & 0.92 & 0.83 & 0.78 & 0.73 & 0.70 & 0.96 & 0.93 & 0.90 & 0.87 & 0.82 & 0.97 & 0.95 & 0.92 & 0.90 & 0.87 & 0.38 \\ 
   & FPR & 0.92 & 0.52 & 0.32 & 0.21 & 0.16 & 0.80 & 0.65 & 0.53 & 0.42 & 0.34 & 0.95 & 0.90 & 0.85 & 0.81 & 0.76 & 0.24 \\ 
   & FDR & 0.02 & 0.04 & 0.05 & 0.07 & 0.07 & 0.01 & 0.02 & 0.03 & 0.03 & 0.04 & 0.01 & 0.01 & 0.02 & 0.02 & 0.03 & 0.15 \\ 
  \multirow{3}{*}{(B)} & TPR & 1.00 & 0.99 & 0.98 & 0.94 & 0.90 & 1.00 & 0.99 & 0.97 & 0.92 & 0.85 & 0.96 & 0.91 & 0.87 & 0.83 & 0.79 & 0.66 \\ 
   & FPR & 0.99 & 0.44 & 0.28 & 0.20 & 0.15 & 0.63 & 0.44 & 0.33 & 0.25 & 0.20 & 0.94 & 0.87 & 0.82 & 0.77 & 0.72 & 0.35 \\ 
   & FDR & 0.00 & 0.00 & 0.01 & 0.03 & 0.05 & 0.00 & 0.00 & 0.02 & 0.04 & 0.07 & 0.02 & 0.04 & 0.06 & 0.08 & 0.10 & 0.16 \\
\bottomrule
\end{tabular}
}
\end{center}

\label{sim-TPR-FPR}
\end{table}

\section{Real data example}\label{sec:5}

We apply the proposed method to the yeast gene expression data \citep{gasch2000genomic}. Following \cite{hirose2017robust}, we focus on analyzing $p=8$ genes involved in galactose utilization \citep{ideker2001integrated}. The sample size is $n=136$. Conducting principal component analysis, we can see 11 outliers in this dataset (see also \cite{finegold2011robust, hirose2017robust}). \cite{hirose2017robust} mentioned that there are two additional outliers, and they consider 13 outliers. Following their study, we compared the two datasets with / without outliers and normalized the data with the median and median absolute deviation (MAD). 

We deal with the data removed for the 13 outliers as a clean dataset (the sample size is $123$), and regard the corresponding estimates as the true estimates for convenience. We compared the proposed method with two frequentist methods presented in Section \ref{sec:4}. The tuning parameter $\lambda$ was selected so that the number of edges is 9 as in \cite{hirose2017robust}. Note that the BR method determines whether the value is zero or not using posterior probability in Section \ref{subsec:2.5} and 6000 posterior samples were drawn by the WBB algorithm.
The results are reported in Figures~\ref{example1-solutionpath} and \ref{example1-posterior}. 
In Figure \ref{example1-solutionpath}, there are no differences between the left panels (without outliers) and right (with outliers) panels for the estimated graphs based on the BR and FR methods, while the FR method is affected by outliers due to non-robustness. When comparing the BR method and the FR method, only one edge is estimated differently. However, the result is not strange because the variable selection method for the two methods is different, as the left panel of Figure \ref{example1-posterior} also clearly illustrates this. It shows a scatter plot of posterior probabilities $\mathrm{P}(|\omega_{ij}| < \varepsilon \mid \mathbf{Y})$ under the BR method and the relationships between the BR and FR methods with colors indicating whether the estimates based on the FR method are zero or not. The colors cross near the posterior probability 50\%. Therefore, it shows that the differing estimate between the two methods arises naturally from the ambiguity of the posterior probability.
The posterior probabilities for the included outliers in the data are similar to those of the data without outliers because the scatter plot is close to the straight line ($y=x$). The 95\% credible intervals (CI) for the BR methods are also shown in the right panel of Figure \ref{example1-posterior}. The CIs for nonzero estimates based on posterior probability correspond to dark-colored bars. The posterior means and CIs for these bars are relatively far from zero. From the posterior probabilities and the credible intervals in Figure \ref{example1-posterior}, it seems that the posterior robustness holds, and the BR method gives the uncertainty for the shrinkage through the posterior probability and credible intervals. Additional information on the real data example including the sensitivity analysis of the tuning parameter $\lambda$ and the result using the BG method is reported in the Supplementary Materials.

\begin{figure}[htbp]
\begin{center}
\includegraphics[width=\linewidth]{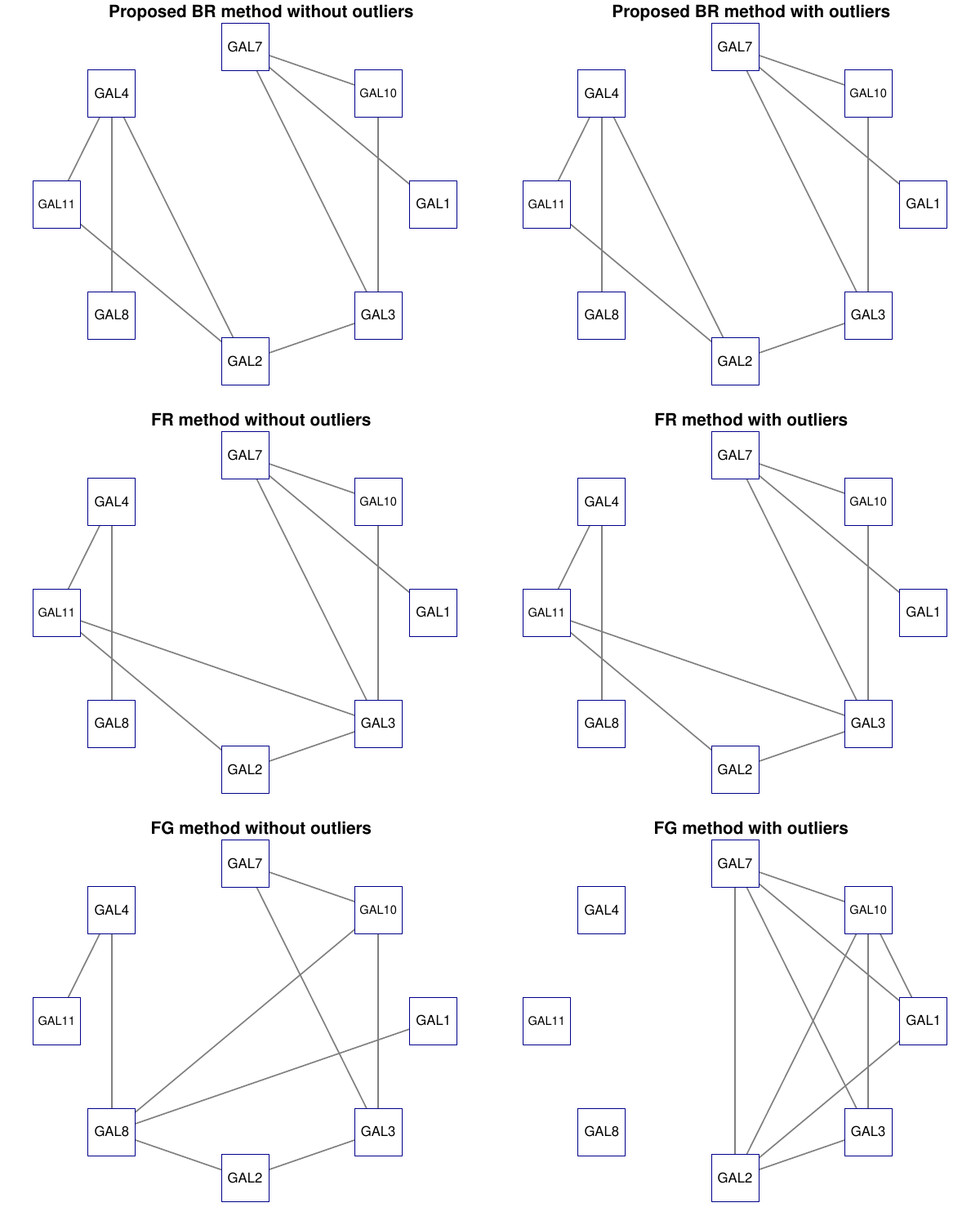}
\caption{Estimated graphs via the BR, FR, and FG methods from top to bottom. The right panels are based on the data without 13 outliers, and the left panels are based on all data with outliers.}
\label{example1-solutionpath}
\end{center}
\end{figure}

\begin{figure}[htbp]
\begin{center}
\includegraphics[width=\linewidth]{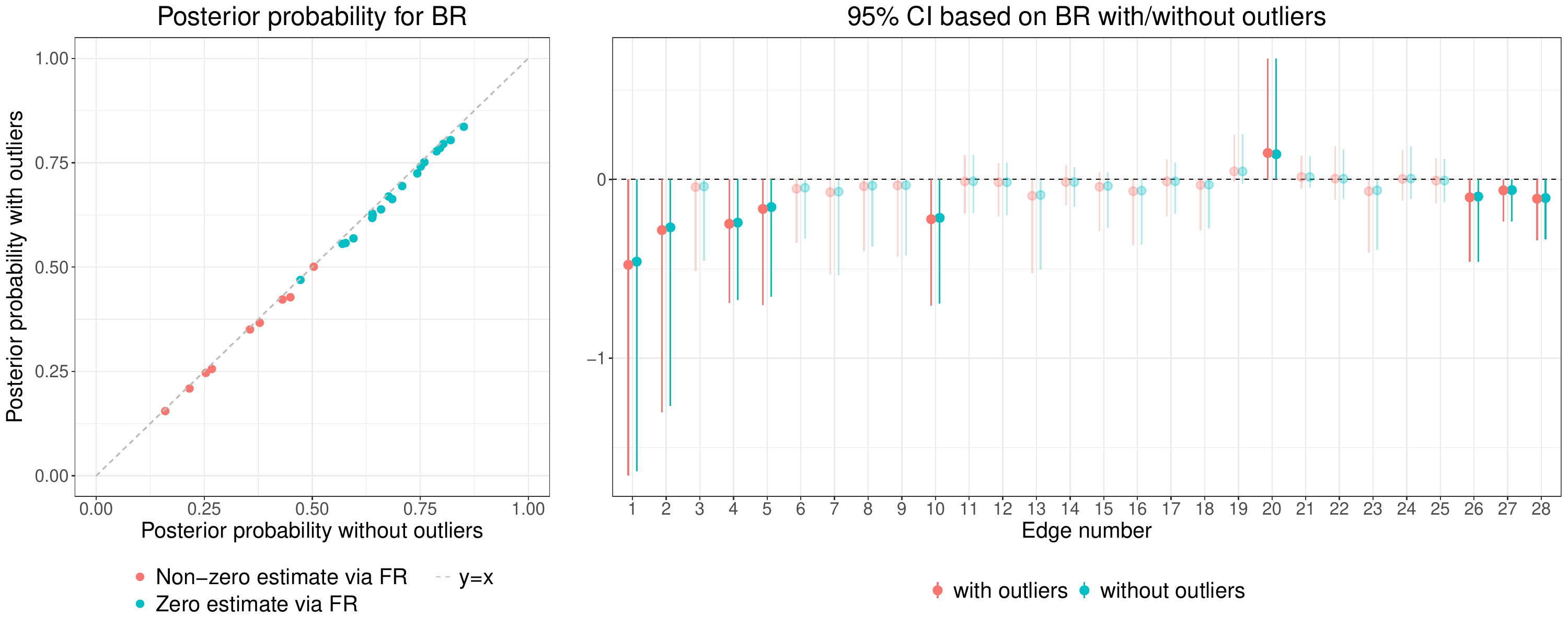}
\caption{The comparison of the posterior probabilities and 95\% credible intervals based on the data with/without outliers. In the right panel, the dark-colored bars are non-zero elements based on posterior probability, and the corresponding edges are drawn in Figure \ref{example1-solutionpath}.}
\label{example1-posterior}
\end{center}
\end{figure}

\section{Concluding remarks}\label{sec:6}

A robust Bayesian graphical lasso based on the $\gamma$-divergence was proposed. 
The proposed posterior distribution was constructed by matching the MAP estimate with the frequentist estimate by \cite{hirose2017robust}. 
The proposed model was shown to have a desirable robustness property, called posterior robustness. 
There are some future works as follows. First, although we focused on using the Laplace prior in this paper, the proposed method can be extended to other types of shrinkage priors. Since the proposed MAP $\gamma$-posterior under popular spike-and-slab (e.g., \cite{wang2015scaling}) and horseshoe (e.g., \cite{li2019graphical}) priors is proper by carefully selecting priors for the diagonal elements (Theorem \ref{posterior-proper}), the posterior robustness still holds. However, the approximation of the posterior distribution is not straightforward. 
Second, in the proposed optimization-based sampling algorithm based on the WBB, the random weight vector distributed as the Dirichlet distribution (see Section \ref{subsec:2.4}) is interpretable as the parameter that controls the spread of the posterior distribution. Hence, by calibrating the hyper-parameters in the Dirichlet distribution, we may derive valid credible intervals in the presence of outliers. 

\section*{Acknowledgments}

The authors would like to thank an Associate Editor and two anonymous referees for their valuable comments and suggestions to improve the quality of this article.
This work was supported by JST, the establishment of university fellowships toward the creation of science technology innovation, Grant Number JPMJFS2129. This work is partially supported by the Japan Society for the Promotion of Science (grant number: 21K13835). The authors also thank Professor Kei Hirose of Kyushu University for providing real data in Section \ref{sec:5}.

\appendix

\section{Proofs of propositions}\label{sec:A}
In this section, we give the proofs for three propositions. 

\begin{proof}[\textbf{\upshape Proof of Proposition \ref{KL-robust}:}]
The ratio of the posterior densities is expressed by
\begin{align*}
\frac{\pi_{\mathrm{KL}}(\mathbf{\Omega}\mid \mathcal{D})}{\pi_{\mathrm{KL}}(\mathbf{\Omega}\mid \mathcal{D}^*)}&=\frac{p(\mathcal{D}^*)}{p(\mathcal{D})}
\frac{\prod_{i=1}^n|\mathbf{\Omega}|^{1/2}\exp\left(-\bm{y}_i^{\top}\mathbf{\Omega} \bm{y}_i/2\right)}{\prod_{i\in\mathcal{K}}|\mathbf{\Omega}|^{1/2}\exp\left(-\bm{y}_i^{\top}\mathbf{\Omega} \bm{y}_i/2\right)}\\
&=\frac{p(\mathcal{D}^*)}{p(\mathcal{D})}
\prod_{i\in\mathcal{L}}|\mathbf{\Omega}|^{1/2}\exp\left(-\bm{y}_i^{\top}\mathbf{\Omega} \bm{y}_i/2\right),
\end{align*}
where 
\begin{align*}
p(\mathcal{D})&=\int \prod_{i=1}^n|\mathbf{\Omega}|^{1/2}\exp\left(-\bm{y}_i^{\top}\mathbf{\Omega} \bm{y}_i/2\right)\pi(\mathbf{\Omega}) d\mathbf{\Omega},\\
p(\mathcal{D}^*)&=\int  \prod_{i\in\mathcal{K}}|\mathbf{\Omega}|^{1/2}\exp\left(-\bm{y}_i^{\top}\mathbf{\Omega} \bm{y}_i/2\right)\pi(\mathbf{\Omega}) d\mathbf{\Omega}.
\end{align*}
Since $\bm{y}_i\neq \bm{0}$ for $i\in\mathcal{L}$, we obtain
\begin{align*}
p(\mathcal{D})&=\int \prod_{i=1}^n|\mathbf{\Omega}|^{1/2}\exp\left(-\bm{y}_i^{\top}\mathbf{\Omega} \bm{y}_i/2\right)\pi(\mathbf{\Omega}) d\mathbf{\Omega}\\
&=\int \prod_{i\in\mathcal{L}}|\mathbf{\Omega}|^{1/2}\exp\left(-\bm{y}_i^{\top}\mathbf{\Omega} \bm{y}_i/2\right) \prod_{i\in\mathcal{K}}|\mathbf{\Omega}|^{1/2}\exp\left(-\bm{y}_i^{\top}\mathbf{\Omega} \bm{y}_i/2\right)\pi(\mathbf{\Omega}) d\mathbf{\Omega}\\
&\le C\int \prod_{i\in\mathcal{K}}|\mathbf{\Omega}|^{1/2}\exp\left(-\bm{y}_i^{\top}\mathbf{\Omega} \bm{y}_i/2\right)\pi(\mathbf{\Omega}) d\mathbf{\Omega}=Cp(D^*),
\end{align*}
where $C$ is a constant. Hence, we have 
\begin{align*}
\frac{\pi_{\mathrm{KL}}(\mathbf{\Omega}\mid \mathcal{D})}{\pi_{\mathrm{KL}}(\mathbf{\Omega}\mid \mathcal{D}^*)}&=\frac{p(\mathcal{D}^*)}{p(\mathcal{D})}
\prod_{i\in\mathcal{L}}|\mathbf{\Omega}|^{1/2}\exp\left(-\bm{y}_i^{\top}\mathbf{\Omega} \bm{y}_i/2\right)\\
&\ge \frac{p(\mathcal{D}^*)}{Cp(\mathcal{D}^*)}
\prod_{i\in\mathcal{L}}|\mathbf{\Omega}|^{1/2}\exp\left(-\bm{y}_i^{\top}\mathbf{\Omega} \bm{y}_i/2\right)\\
&=\frac{1}{C}\prod_{i\in\mathcal{L}}|\mathbf{\Omega}|^{1/2}\exp\left(-\bm{y}_i^{\top}\mathbf{\Omega} \bm{y}_i/2\right)
\end{align*}
and 
\begin{align*}
\lim_{z\to\infty}\frac{\pi_{\mathrm{KL}}(\mathbf{\Omega}\mid \mathcal{D})}{\pi_{\mathrm{KL}}(\mathbf{\Omega}\mid \mathcal{D}^*)}&=\lim_{z\to\infty}\frac{1}{C}\prod_{i\in\mathcal{L}}|\mathbf{\Omega}|^{1/2}\exp\left(-\bm{y}_i^{\top}\mathbf{\Omega} \bm{y}_i/2\right)=0.
\end{align*}
Using the same argument as Theorem \ref{gamma-robust}, we have
\begin{align*}
\lim_{z\to\infty}\int|\pi_{\mathrm{KL}}(\mathbf{\Omega}\mid \mathcal{D})-\pi_{\mathrm{KL}}(\mathbf{\Omega}\mid \mathcal{D}^*)|d\mathbf{\Omega}=1 \ne 0.
\end{align*}
This completes the proof. 
\end{proof}

\begin{proof}[\textbf{\upshape Proof of Proposition \ref{t-robust}:}]
The ratio of the posterior densities is expressed by
\begin{align*}
\frac{\pi_t(\mathbf{\Omega}\mid \mathcal{D})}{\pi_t(\mathbf{\Omega}\mid \mathcal{D}^*)}&=\frac{p(\mathcal{D}^*)}{p(\mathcal{D})}
\frac{\prod_{i=1}^n|\mathbf{\Omega}|^{1/2}\left[1+\bm{y}_i^{\top}\mathbf{\Omega} \bm{y}_i/(\nu-2)\right]^{-(\nu+p)}}{\prod_{i\in\mathcal{K}}|\mathbf{\Omega}|^{1/2}\left[1+\bm{y}_i^{\top}\mathbf{\Omega} \bm{y}_i/(\nu-2)\right]^{-(\nu+p)}}\\
&=\frac{p(\mathcal{D}^*)}{p(\mathcal{D})}
\prod_{i\in\mathcal{L}}|\mathbf{\Omega}|^{1/2}\left[1+\bm{y}_i^{\top}\mathbf{\Omega} \bm{y}_i/(\nu-2)\right]^{-(\nu+p)},
\end{align*}
where 
\begin{align*}
p(\mathcal{D})&=\int \prod_{i=1}^n|\mathbf{\Omega}|^{1/2}\left[1+\bm{y}_i^{\top}\mathbf{\Omega} \bm{y}_i/(\nu-2)\right]^{-(\nu+p)}\pi(\mathbf{\Omega}) d\mathbf{\Omega},\\
p(\mathcal{D}^*)&=\int  \prod_{i\in\mathcal{K}}|\mathbf{\Omega}|^{1/2}\left[1+\bm{y}_i^{\top}\mathbf{\Omega} \bm{y}_i/(\nu-2)\right]^{-(\nu+p)}\pi(\mathbf{\Omega}) d\mathbf{\Omega}.
\end{align*}
Since $\bm{y}_i\neq \bm{0}$ for $i\in\mathcal{L}$, we obtain
\begin{align*}
p(\mathcal{D})&=\int \prod_{i=1}^n|\mathbf{\Omega}|^{1/2}\left[1+\bm{y}_i^{\top}\mathbf{\Omega} \bm{y}_i/(\nu-2)\right]^{-(\nu+p)}\pi(\mathbf{\Omega}) d\mathbf{\Omega}\\
&=\int \prod_{i\in\mathcal{L}}|\mathbf{\Omega}|^{1/2}\left[1+\bm{y}_i^{\top}\mathbf{\Omega} \bm{y}_i/(\nu-2)\right]^{-(\nu+p)} \prod_{i\in\mathcal{K}}|\mathbf{\Omega}|^{1/2}\left[1+\bm{y}_i^{\top}\mathbf{\Omega} \bm{y}_i/(\nu-2)\right]^{-(\nu+p)}\pi(\mathbf{\Omega}) d\mathbf{\Omega}\\
&\le C\int \prod_{i\in\mathcal{K}}|\mathbf{\Omega}|^{1/2}\left[1+\bm{y}_i^{\top}\mathbf{\Omega} \bm{y}_i/(\nu-2)\right]^{-(\nu+p)}\pi(\mathbf{\Omega}) d\mathbf{\Omega}=Cp(\mathcal{D}^*),
\end{align*}
where $C$ is a constant. Then we have
\begin{align*}
\frac{\pi_t(\mathbf{\Omega}\mid \mathcal{D})}{\pi_t(\mathbf{\Omega}\mid \mathcal{D}^*)}&=\frac{p(\mathcal{D}^*)}{p(\mathcal{D})}
\prod_{i\in\mathcal{L}}|\mathbf{\Omega}|^{1/2}\left[1+\bm{y}_i^{\top}\mathbf{\Omega} \bm{y}_i/(\nu-2)\right]^{-(\nu+p)}\\
&\ge \frac{p(\mathcal{D}^*)}{Cp(\mathcal{D}^*)}
\prod_{i\in\mathcal{L}}|\mathbf{\Omega}|^{1/2}\left[1+\bm{y}_i^{\top}\mathbf{\Omega} \bm{y}_i/(\nu-2)\right]^{-(\nu+p)}\\
&=\frac{1}{C}\prod_{i\in\mathcal{L}}|\mathbf{\Omega}|^{1/2}\left[1+\bm{y}_i^{\top}\mathbf{\Omega} \bm{y}_i/(\nu-2)\right]^{-(\nu+p)}
\end{align*}
and 
\begin{align*}
\lim_{z\to\infty}\frac{\pi_t(\mathbf{\Omega}\mid \mathcal{D})}{\pi_t(\mathbf{\Omega}\mid \mathcal{D}^*)}&=\lim_{z\to\infty}\frac{1}{C}\prod_{i\in\mathcal{L}}|\mathbf{\Omega}|^{1/2}\left[1+\bm{y}_i^{\top}\mathbf{\Omega} \bm{y}_i/(\nu-2)\right]^{-(\nu+p)}=0.
\end{align*}
Using the same argument as Theorem \ref{gamma-robust}, we have
\begin{align*}
\lim_{z\to\infty}\int|\pi_t(\mathbf{\Omega}\mid \mathcal{D})-\pi_t(\mathbf{\Omega}\mid \mathcal{D}^*)|d\mathbf{\Omega}=1 \ne 0.
\end{align*}
This completes the proof. 
\end{proof}

\begin{proof}[\textbf{\upshape Proof of Proposition \ref{DP-robust}:}]
The ratio of the posterior densities is expressed by
\begin{align*}
\frac{\pi_{\mathrm{DP}}(\mathbf{\Omega}\mid \mathcal{D})}{\pi_{\mathrm{DP}}(\mathbf{\Omega}\mid \mathcal{D}^*)}&=\frac{p(\mathcal{D}^*)}{p(\mathcal{D})}
\frac{\exp\left(Q_n^{(\alpha)}(\mathcal{D}\mid \mathbf{\Omega})\right)}{\exp\left(Q_n^{(\alpha)}(\mathcal{D}^*\mid \mathbf{\Omega})\right)},
\end{align*}
where
\begin{align*}
Q_n^{(\alpha)}(\mathcal{D}\mid \mathbf{\Omega})&=(2\pi)^{-\alpha p/2}|\mathbf{\Omega}|^{\alpha/2}\left[\frac{1}{\alpha}\sum_{i=1}^n\exp\left(-\frac{\alpha}{2}\bm{y}_i^{\top}\mathbf{\Omega} \bm{y}_i\right)-n(1+\alpha)^{1-\alpha/2}\right],\\
Q_n^{(\alpha)}(\mathcal{D}^*\mid \mathbf{\Omega})&=(2\pi)^{-\alpha p/2}|\Omega|^{\alpha/2}\left[\frac{1}{\alpha}\sum_{i\in\mathcal{K}}\exp\left(-\frac{\alpha}{2}\bm{y}_i^{\top}\mathbf{\Omega} \bm{y}_i\right)-|\mathcal{K}|(1+\alpha)^{1-\alpha/2}\right],\\
p(\mathcal{D})&=\int \exp\left(Q_n^{(\alpha)}(\mathcal{D}\mid \mathbf{\Omega})\right)\pi(\mathbf{\Omega}) d\mathbf{\Omega},\\
p(\mathcal{D}^*)&=\int  \exp\left(Q_n^{(\alpha)}(\mathcal{D}^*\mid \mathbf{\Omega})\right)\pi(\mathbf{\Omega}) d\mathbf{\Omega}.
\end{align*}
Then we obtain
\begin{align*}
\lim_{z\to\infty}\exp\left(Q_n^{(\alpha)}(\mathcal{D}\mid \mathbf{\Omega})\right)
&=\lim_{z\to\infty}\exp\left[(2\pi)^{-\alpha p/2}|\mathbf{\Omega}|^{\alpha/2}\left\{\frac{1}{\alpha}\sum_{i=1}^n\exp\left(-\frac{\alpha}{2}\bm{y}_i^{\top}\mathbf{\Omega} \bm{y}_i\right)-n(1+\alpha)^{1-\alpha/2}\right\}\right]\\
&=\exp\left[(2\pi)^{-\alpha p/2}|\mathbf{\Omega}|^{\alpha/2}\left\{\frac{1}{\alpha}\sum_{i\in\mathcal{K}}\exp\left(-\frac{\alpha}{2}\bm{y}_i^{\top}\mathbf{\Omega} \bm{y}_i\right)-n(1+\alpha)^{1-\alpha/2}\right\}\right]\\
&=\exp\left[(2\pi)^{-\alpha p/2}|\mathbf{\Omega}|^{\alpha/2}\frac{1}{\alpha}\sum_{i\in\mathcal{K}}\exp\left(-\frac{\alpha}{2}\bm{y}_i^{\top}\mathbf{\Omega} \bm{y}_i\right)\right]\\
&\quad \times \exp\left[-(2\pi)^{-\alpha p/2}n(1+\alpha)^{1-\alpha/2}|\mathbf{\Omega}|^{\alpha/2}\right]\\
&=\exp\left(Q_n^{(\alpha)}(\mathcal{D}^*\mid \mathbf{\Omega})\right)\exp\left[-(2\pi)^{-\alpha p/2}(n-|\mathcal{K}|)(1+\alpha)^{1-\alpha/2}|\mathbf{\Omega}|^{\alpha/2}\right].
\end{align*}
From Lebesgue's dominated convergence theorem, it holds that 
\begin{align*}
\lim_{z\to\infty} p(\mathcal{D})&=\lim_{z\to\infty} \int  \exp\left(Q_n^{(\alpha)}(\mathcal{D}\mid \mathbf{\Omega})\right)\pi(\mathbf{\Omega})d\mathbf{\Omega}\\
&=\int \lim_{z\to\infty} \exp\left(Q_n^{(\alpha)}(\mathcal{D}\mid \mathbf{\Omega})\right)\pi(\mathbf{\Omega})d\mathbf{\Omega}\\
&=\int \exp\left(Q_n^{(\alpha)}(\mathcal{D}^*\mid \mathbf{\Omega})\right)\exp\left[-(2\pi)^{-\alpha p/2}(n-|\mathcal{K}|)(1+\alpha)^{1-\alpha/2}|\mathbf{\Omega}|^{\alpha/2}\right]\pi(\mathbf{\Omega})d\mathbf{\Omega}\\
&=C_{\mathcal{D}^*}>0,
\end{align*}
where $C_{\mathcal{D}^*}$ is a constant that does not depend on $\mathbf{\Omega}$, and we have $\lim_{z\to\infty} p(\mathcal{D}^*)/p(\mathcal{D})=p(\mathcal{D}^*)/C_{\mathcal{D}^*}$. Therefore, the limit of the ratio is calculated by
\begin{align*}
\lim_{z\to\infty}\frac{\pi_{\mathrm{DP}}(\mathbf{\Omega}\mid \mathcal{D})}{\pi_{\mathrm{DP}}(\mathbf{\Omega}\mid \mathcal{D}^*)}&=\frac{p(\mathcal{D}^*)}{C_{\mathcal{D}^*}}\exp\left[-(2\pi)^{-\alpha p/2}(n-|\mathcal{K}|)(1+\alpha)^{1-\alpha/2}|\mathbf{\Omega}|^{\alpha/2}\right]\neq 1.
\end{align*}
This completes the proof.
\end{proof}

\section{The details of algorithms}\label{sec:B}

We summarized the details of two algorithms: weighted Bayesian bootstrap via MM algorithm, and Gibbs sampler for $t$-likelihood.

\subsection{Derivation of \eqref{bound}}\label{subsec:B1}

Considering the lasso type prior
\begin{align*}
\pi(\mathbf{\Omega})&\propto\exp(-\lambda\|\mathbf{\Omega}\|_1)1_{\{\mathbf{\Omega}\in M^+\}}=\prod_{i=1}^p\mathrm{Exp}(w_{ii}\mid \lambda)\prod_{i<j}^p\mathrm{Lap}(w_{ij}\mid \lambda)1_{\{\mathbf{\Omega}\in M^+\}},
\end{align*}
the weighted objective function is given by
\begin{align*}
L_{\bm{w}}(\mathbf{\Omega})&=-\frac{1}{\gamma}\log\left[\frac{1}{n}\sum_{i=1}^nw_i\frac{f(\bm{y}_i\mid \mathbf{\Omega})^{\gamma}}{\int f(\bm{y}\mid \mathbf{\Omega})^{1+\gamma}dy}\right]+\lambda w_0\|\mathbf{\Omega}\|_1\\
&=-\frac{1}{\gamma}\log\left\{\frac{1}{n}\sum_{i=1}^nw_if(\bm{y}_i\mid \mathbf{\Omega})^{\gamma}\right\}+\frac{\gamma}{2(1+\gamma)}\log|\mathbf{\Omega}|+\lambda w_0\|\mathbf{\Omega}\|_1,
\end{align*}
where $\bm{w}=(w_0,w_1,\dots,w_n)^\top$ is a weight defined in Subsection \ref{subsec:2.4}.
By using Jensen's inequality, the first term is bounded by
\begin{align*}
-\frac{1}{\gamma}\log\left\{\frac{1}{n}\sum_{i=1}^nw_if(\bm{y}_i\mid \mathbf{\Omega})^{\gamma}\right\}
\le C_1-\frac{1}{2}\log |\mathbf{\Omega}|+\frac{1}{2}\sum_{i=1}^ns_i^* \bm{y}_i^{\top}\mathbf{\Omega} \bm{y}_i,
\end{align*}
where $s_i^*=w_if(\bm{y}_i\mid\mathbf{\Omega})^{\gamma}/\sum_{j=1}^nw_jf(\bm{y}_j\mid\mathbf{\Omega})^{\gamma}$ and $C_1$ is a constant. 
Therefore, the objective function $L_{\bm{w}}(\mathbf{\Omega})$ is bounded by
\begin{align*}
L_{\bm{w}}(\mathbf{\Omega})&=-\frac{1}{\gamma}\log\left\{\sum_{i=1}^nw_if(\bm{y}_i\mid \mathbf{\Omega})^{\gamma}\right\}+\frac{\gamma}{2(1+\gamma)}\log|\mathbf{\Omega}|+\lambda w_0\|\mathbf{\Omega}\|_1\\
&\le \frac{1}{2(1+\gamma)}\left[ \mathrm{tr}\left\{\left((1+\gamma)\sum_{i=1}^ns_i^*\bm{y}_i\bm{y}_i^{\top}\right)\mathbf{\Omega}\right\}-\log|\mathbf{\Omega}|+2(1+\gamma)\lambda w_0\|\mathbf{\Omega}\|_1\right]\\
&\propto \mathrm{tr}\left\{\mathbf{S}^*\mathbf{\Omega}\right\}-\log|\mathbf{\Omega}|+\rho\|\mathbf{\Omega}\|_1,
\end{align*}
where
\begin{align*}
\mathbf{S}^*&=(1+\gamma)\sum_{i=1}^ns_i^*\bm{y}_i\bm{y}_i^{\top},\quad \rho=2(1+\gamma)\lambda w_0.
\end{align*}

\subsection{MCMC algorithm for Bayesian $t$-graphical lasso}\label{subsec:B2}

The joint probability density function of the multivariate $t$-distribution is given by
\begin{equation*}
p(\mathbf{Y}\mid \mathbf{\Sigma})\propto \prod_{i=1}^n p(\bm{y}_i\mid\mathbf{\Sigma})\propto \prod_{i=1}^n |\mathbf{\Sigma}|^{-1/2}\left(1+\frac{1}{\nu}\bm{y}_i^{\top}\mathbf{\Sigma}^{-1} \bm{y}_i\right),
\end{equation*}
where $\nu$ is the degrees of freedom and the covariance matrix is $\{\nu/(\nu-2)\} \mathbf{\Sigma}$ for $\nu>2$. By using the Gaussian scale mixture representation of the $t$-distribution, the conditional density is given by
\begin{align*}
\bm{y}_i\mid \mathbf{\Sigma}, u\sim \mathcal{N}_p(\bm{0},\mathbf{\Sigma}/u),\quad
u\sim \mathrm{Ga}(\nu/2,\nu/2)
\end{align*}
We define the precision matrix $\mathbf{\Omega}=( \{\nu/(\nu-2)\}\mathbf{\Sigma} )^{-1}$, and then the density functions are rewritten as
\begin{align*}
p(\mathbf{Y}\mid \mathbf{\Omega}, u)&\propto u^{np/2}|\mathbf{\Omega}|^{n/2}\exp\left(-\frac{\nu u}{2(\nu-2)}\sum_{i=1}^n \bm{y}_i^{\top}\mathbf{\Omega} \bm{y}_i\right),\\
p(u)&\propto u^{-1+\nu/2}\exp\left(-\frac{\nu}{2}u\right).
\end{align*}
For the prior of $\mathbf{\Omega}$, we assume the Laplace and exponential priors:
\begin{align*}
p(\mathbf{\Omega}\mid \tau, \lambda)&\propto \prod_{i<j}\left[\tau_{ij}^{-1/2}\exp\left(-\frac{\omega_{ij}^2}{2\tau_{ij}}\right)\right]
\prod_{i=1}^p\left[\exp\left(-\frac{\lambda}{2}\omega_{ij}\right)\right]1_{\mathbf{\Omega}\in M^+},\\
p(\tau\mid \lambda)&\propto \prod_{i<j}\left[\exp\left(-\frac{\lambda^2}{2}\tau_{ij}\right)\right].
\end{align*}
From this formulation, we can construct a Gibbs sampler based on block Gibbs sampler \citep{wang2012bayesian}. The full conditional distribution of $u$ is given by
\begin{align*}
p(u\mid y, \mathbf{\Omega})\sim \mathrm{Ga}\left(\frac{np+\nu}{2}, \frac{\nu}{2(\nu-2)}\sum_{i=1}^n\bm{y}_i^{\top}\mathbf{\Omega} \bm{y}_i+\frac{\nu}{2}\right).
\end{align*}

\section{Additional simulation information}\label{sec:C}

\subsection{True precision matrix}\label{subsec:C1}

The true precision matrix (A) in the simulation study is as follows: 

{\small
\begin{align*}
\left(
\begin{array}{cccccccccccc}
0.239 & 0.117 &  &  &  &  &  & 0.031 &  &  &  &  \\ 
0.117 & 1.554 &  &  &  &  &  &  &  &  &  &  \\ 
  &  & 0.362 & 0.002 &  &  &  &  &  &  &  &  \\ 
  &  & 0.002 & 0.199 & 0.094 &  &  &  &  &  &  &  \\ 
  &  &  & 0.094 & 0.349 &  &  &  &  &  &  & -0.036 \\ 
  &  &  &  &  & 0.295 & -0.229 & 0.002 &  &  &  &  \\ 
  &  &  &  &  & -0.229 & 0.715 &  &  &  &  &  \\ 
 0.031 &  &  &  &  & 0.002 &  & 0.164 & 0.112 & -0.028 & -0.008 &  \\ 
  &  &  &  &  &  &  & 0.112 & 0.518 & -0.193 & -0.09 &  \\ 
  &  &  &  &  &  &  & -0.028 & -0.193 & 0.379 & 0.167 &  \\ 
  &  &  &  &  &  &  & -0.008 & -0.09 & 0.167 & 0.159 &  \\ 
  &  &  &  & -0.036 &  &  &  &  &  &  & 0.207 \\ 
 \end{array}
 \right)
\end{align*}
}

\subsection{The difference between $\gamma=0.05$ and $\gamma=0.1$ for the robust methods}\label{subsec:C2}

First, to compare the results of $\gamma=0.05$ with that of $\gamma=0.01$, we report the simulation results for the BR and FR methods for $\gamma=0.05$. The results of RMSE are summarized in Figure \ref{sim-RMSE2}. For comparison, we also show the results for $\gamma=0.1$ in the same figures. For scenario (a) (non-outliers case), the results for $\gamma=0.05$ are similar to those for $\gamma=0.1$. However, $\gamma=0.1$ tends to provide more robust estimates for the scenario in which the value of the outliers is large. Next, we show the effect on the average length (AL) and the coverage probability (CP) in Table~\ref{sim-AL-CP-gamma=005-01}. CP for $\gamma=0.05$ is farther away from 0.95 than that for $\gamma=0.1$ and is affected by outliers. AL of $\gamma=0.05$ is similar to $\gamma=0.1$ except scenario (c) $\eta=20$ and scenario (c-B) $\eta=10$. Since $\gamma$ plays the role of the robustness of the posterior distribution, the posterior for each $\gamma$ may change depending on the data such as outlier levels. 
AL is small due to strong shrinkage since $\lambda$ is larger. $\lambda$ controls the level of shrinkage and the degree of posterior concentration at zero. Therefore, the posterior distribution and its credible interval are also affected by choosing $\lambda$. Constructing a data-dependent calibration method of credible intervals is an interesting future work (see, e.g., \cite{syring2019calibrating, onizuka2024fast}).

\begin{figure}
\begin{center}
\includegraphics[width=\linewidth]{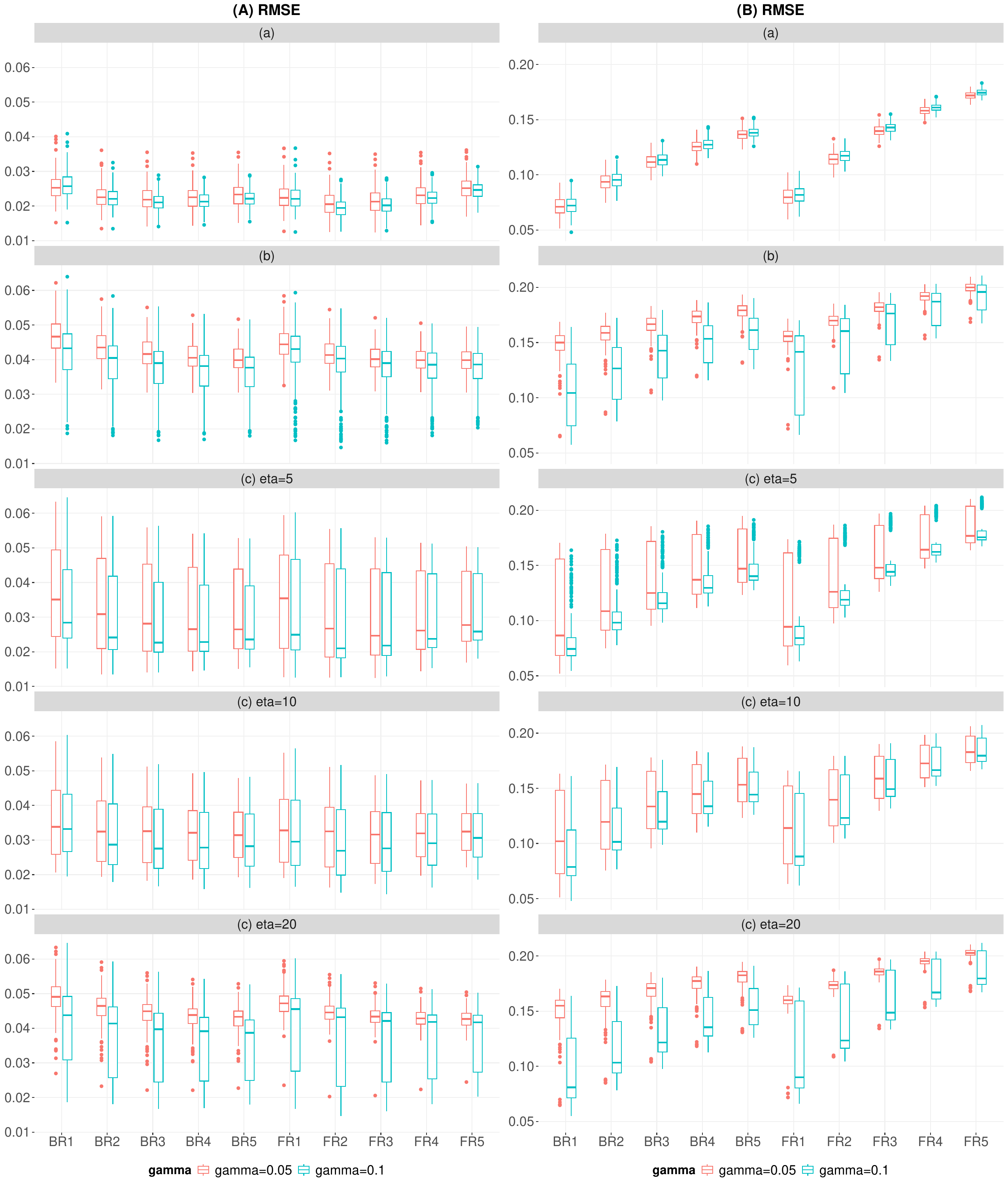}
\caption{The boxplots of RMSE based on 100 repetitions for all scenarios are summarized. The left and right panels correspond to the results for (A) and (B), respectively, and the results for the five data-generating processes are summarized from top to bottom.}
\label{sim-RMSE2}
\end{center}
\end{figure}

\begin{table}[htbp]

\caption{Coverage probabilities and average lengths of 95\% credible intervals averaged over 100 Monte Carlo replications.}
\begin{center}
\resizebox{1.0\textwidth}{!}{ 
\begin{tabular}{cc|cc|cc|cc|cc|cc}
\toprule
  \multicolumn{12}{c}{Data-generating process (a)} \\
  \hline
& &\multicolumn{2}{c|}{BR1} &\multicolumn{2}{c|}{BR2}
&\multicolumn{2}{c|}{BR3} &\multicolumn{2}{c|}{BR4}
&\multicolumn{2}{c}{BR5}\\
\hline
 & & $\gamma=0.05$ & $\gamma=0.1$ & $\gamma=0.05$ & $\gamma=0.1$  & $\gamma=0.05$ & $\gamma=0.1$  & $\gamma=0.05$ & $\gamma=0.1$  & $\gamma=0.05$ & $\gamma=0.1$ \\ 
\hline
\multirow{2}{*}{(A)} & AL & 0.100 & 0.106 & 0.092 & 0.097 & 0.087 & 0.091 & 0.083 & 0.086 & 0.079 & 0.082 \\
 & CP & 0.923 & 0.929 & 0.904 & 0.911 & 0.881 & 0.888 & 0.859 & 0.865 & 0.838 & 0.842 \\ 
\multirow{2}{*}{(B)} & AL & 0.101 & 0.112 & 0.092 & 0.102 & 0.087 & 0.095 & 0.082 & 0.090 & 0.079 & 0.086 \\ 
 & CP & 0.931 & 0.938 & 0.912 & 0.916 & 0.892 & 0.897 & 0.874 & 0.878 & 0.853 & 0.862 \\
 \midrule
 \multicolumn{12}{c}{Data-generating process (b)} \\
  \hline
& &\multicolumn{2}{c|}{BR1} &\multicolumn{2}{c|}{BR2}
&\multicolumn{2}{c|}{BR3} &\multicolumn{2}{c|}{BR4}
&\multicolumn{2}{c}{BR5}\\
\hline
 & & $\gamma=0.05$ & $\gamma=0.1$ & $\gamma=0.05$ & $\gamma=0.1$  & $\gamma=0.05$ & $\gamma=0.1$  & $\gamma=0.05$ & $\gamma=0.1$  & $\gamma=0.05$ & $\gamma=0.1$ \\ 
\hline
\multirow{2}{*}{(A)} & AL & 0.091 & 0.099 & 0.085 & 0.091 & 0.081 & 0.086 & 0.078 & 0.082 & 0.075 & 0.079 \\ 
 & CP & 0.889 & 0.901 & 0.873 & 0.884 & 0.855 & 0.864 & 0.835 & 0.844 & 0.820 & 0.828 \\ 
\multirow{2}{*}{(B)} & AL & 0.091 & 0.107 & 0.085 & 0.098 & 0.081 & 0.092 & 0.077 & 0.087 & 0.075 & 0.084 \\  
 & CP & 0.895 & 0.928 & 0.876 & 0.911 & 0.86 & 0.888 & 0.844 & 0.873 & 0.827 & 0.854 \\
  \midrule
 \multicolumn{12}{c}{Data-generating process (c) $\eta=5$} \\
  \hline
& &\multicolumn{2}{c|}{BR1} &\multicolumn{2}{c|}{BR2}
&\multicolumn{2}{c|}{BR3} &\multicolumn{2}{c|}{BR4}
&\multicolumn{2}{c}{BR5}\\
\hline
 & & $\gamma=0.05$ & $\gamma=0.1$ & $\gamma=0.05$ & $\gamma=0.1$  & $\gamma=0.05$ & $\gamma=0.1$  & $\gamma=0.05$ & $\gamma=0.1$  & $\gamma=0.05$ & $\gamma=0.1$ \\ 
\hline
\multirow{2}{*}{(A)} & AL & 0.093 & 0.107 & 0.086 & 0.098 & 0.081 & 0.092 & 0.078 & 0.087 & 0.075 & 0.084 \\  
 & CP & 0.900 & 0.926 & 0.884 & 0.910 & 0.867 & 0.889 & 0.851 & 0.873 & 0.835 & 0.855 \\ 
\multirow{2}{*}{(B)} & AL & 0.270 & 0.281 & 0.257 & 0.265 & 0.245 & 0.251 & 0.234 & 0.238 & 0.224 & 0.226 \\ 
 & CP & 0.901 & 0.906 & 0.845 & 0.853 & 0.799 & 0.803 & 0.750 & 0.758 & 0.712 & 0.712 \\
  \midrule
 \multicolumn{12}{c}{Data-generating process (c) $\eta=10$} \\
  \hline
& &\multicolumn{2}{c|}{BR1} &\multicolumn{2}{c|}{BR2}
&\multicolumn{2}{c|}{BR3} &\multicolumn{2}{c|}{BR4}
&\multicolumn{2}{c}{BR5}\\
\hline
 & & $\gamma=0.05$ & $\gamma=0.1$ & $\gamma=0.05$ & $\gamma=0.1$  & $\gamma=0.05$ & $\gamma=0.1$  & $\gamma=0.05$ & $\gamma=0.1$  & $\gamma=0.05$ & $\gamma=0.1$ \\ 
\hline
\multirow{2}{*}{(A)} & AL & 0.281 & 0.297 & 0.266 & 0.277 & 0.253 & 0.262 & 0.241 & 0.248 & 0.229 & 0.236 \\ 
 & CP & 0.908 & 0.916 & 0.858 & 0.867 & 0.813 & 0.821 & 0.769 & 0.775 & 0.725 & 0.734 \\ 
\multirow{2}{*}{(B)} & AL & 0.224 & 0.286 & 0.213 & 0.264 & 0.203 & 0.247 & 0.194 & 0.232 & 0.186 & 0.219 \\
 & CP & 0.788 & 0.892 & 0.734 & 0.833 & 0.688 & 0.779 & 0.642 & 0.729 & 0.600 & 0.678 \\
  \midrule
 \multicolumn{12}{c}{Data-generating process (c) $\eta=20$} \\
  \hline
& &\multicolumn{2}{c|}{BR1} &\multicolumn{2}{c|}{BR2}
&\multicolumn{2}{c|}{BR3} &\multicolumn{2}{c|}{BR4}
&\multicolumn{2}{c}{BR5}\\
\hline
 & & $\gamma=0.05$ & $\gamma=0.1$ & $\gamma=0.05$ & $\gamma=0.1$  & $\gamma=0.05$ & $\gamma=0.1$  & $\gamma=0.05$ & $\gamma=0.1$  & $\gamma=0.05$ & $\gamma=0.1$ \\ 
\hline
\multirow{2}{*}{(A)} & AL & 0.234 & 0.289 & 0.221 & 0.268 & 0.211 & 0.252 & 0.201 & 0.238 & 0.191 & 0.226 \\  
 & CP & 0.811 & 0.905 & 0.753 & 0.848 & 0.705 & 0.800 & 0.661 & 0.749 & 0.623 & 0.701 \\ 
\multirow{2}{*}{(B)} & AL & 0.235 & 0.290 & 0.222 & 0.270 & 0.211 & 0.253 & 0.201 & 0.239 & 0.191 & 0.226 \\
 & CP & 0.815 & 0.906 & 0.759 & 0.853 & 0.711 & 0.804 & 0.668 & 0.756 & 0.626 & 0.705 \\  
\bottomrule
\end{tabular}
}
\end{center}

\label{sim-AL-CP-gamma=005-01}
\end{table}

\bibliographystyle{chicago}
\bibliography{refs}

\newpage
\setcounter{page}{1}
\setcounter{equation}{0}
\renewcommand{\theequation}{S\arabic{equation}}
\setcounter{section}{0}
\renewcommand{\thesection}{S\arabic{section}}
\setcounter{table}{0}
\renewcommand{\thetable}{S\arabic{table}}
\setcounter{figure}{0}
\renewcommand{\thefigure}{S\arabic{figure}}

\begin{center}
{\LARGE\bf Supplementary Materials for ``Robust Bayesian graphical modeling using $\gamma$-divergence"}
\end{center}

\begin{center}
{\large Takahiro Onizuka$^1$ and Shintaro Hashimoto$^2$}
\end{center}

\medskip
$^1$ Graduate School of Social Sciences, Chiba University, Japan

$^2$ Department of Mathematics, Hiroshima University, Japan

\vspace{1cm}

This Supplementary Material provides additional simulation results of simulation studies related to the main text.

\section{Selection of tuning parameter}

Although we fixed the tuning parameter $\lambda$ of the method proposed in the main article, we considered the selection of the tuning parameter. In Bayesian methods, to avoid the selection of the hyper-parameter, we often assume a prior distribution for the parameter. Since the proposed optimization algorithm does not have a sampling procedure, it may be difficult to induce the hyperprior and the sampling from the joint posterior with the hyper-prior added. \cite{hashimoto2020robust} proposed a new algorithm in the robust Bayesian method through the combination of optimization and sampling of parameters, and here we adopt their idea. We assume the gamma distribution $\mathrm{Ga}(a,b)$ to the tuning parameter $\lambda$ as a conjugate prior, and we set $a=b=0.1$ as a non-informative prior. The corresponding adaptive sampling method called Weighted Bayesian Bootstrap within Gibbs (WBBG) is shown in Algorithm~\ref{algo:SWBB}.

\begin{algorithm*}[thbp] 
\caption{\bf--- Weighted Bayesian bootstrap within Gibbs.}
\label{algo:SWBB}
\begin{itemize}
\item[1] For fixed $\lambda^{(t-1)}$, calculate the point estimate of precision matrix $\Omega^{(t)}$ through Algorithm 1 in the main manuscript. 
\item[2] Sample $\lambda^{(t)}$ from the full conditional distribution of $\lambda$: $\mathrm{Ga}(a,b+\|\mathbf{\Omega}^{(t)}\|_1)$.
\item[3] Cycle steps 1 and 2, and get the $m$th posterior sample.
\end{itemize}
\end{algorithm*}

The differences between the original algorithm and Algorithm~\ref{algo:SWBB} are as follows:
\begin{itemize}
\item Due to the sampling of $\lambda$, the parallel computation is impossible.
\item $\lambda$ is adaptively selected via MCMC.
\item Since $\lambda$ takes various values in the sampling iteration, the posterior does not have a spike at zero. 
\end{itemize}
The comparison of posterior distributions for each algorithm is shown in Figure~\ref{posterior-BR-SBR-BG}. The posterior of the BR via WBBG has a shape similar to the BG method, which is constructed by Gibbs sampling.
To verify the performance, we compare the WBBG algorithm with the original algorithm through the simulation study in the main manuscript. The result is reported in Tables \ref{Sample-lambda-ab} and \ref{Sample-lambda-c}. Note that the result of the original methods is the same as that of the main manuscript. It seems that the RMSE of the BR method with the WBBG is a reasonable behavior as well as the proposed BR method with the WBB. 
Howver, the TPR and FPR of WBBG are relatively high. This indicates that the selected graph tends to be dense. In fact, Figure S1 shows that the posterior density of $w_{ij}$ under WBBG has a mass around 0 smaller than that of WBB. To select a reasonable (non-dense) graph by WBBG, we have to consider a new criterion for the $\gamma$-divergence based model within Gibbs instead of the median probability criterion in Section 2.5 of the main manuscript. This is the current limitation of this approach and an interesting future work.

\begin{figure}[htbp]
\begin{center}
\includegraphics[width=\linewidth]{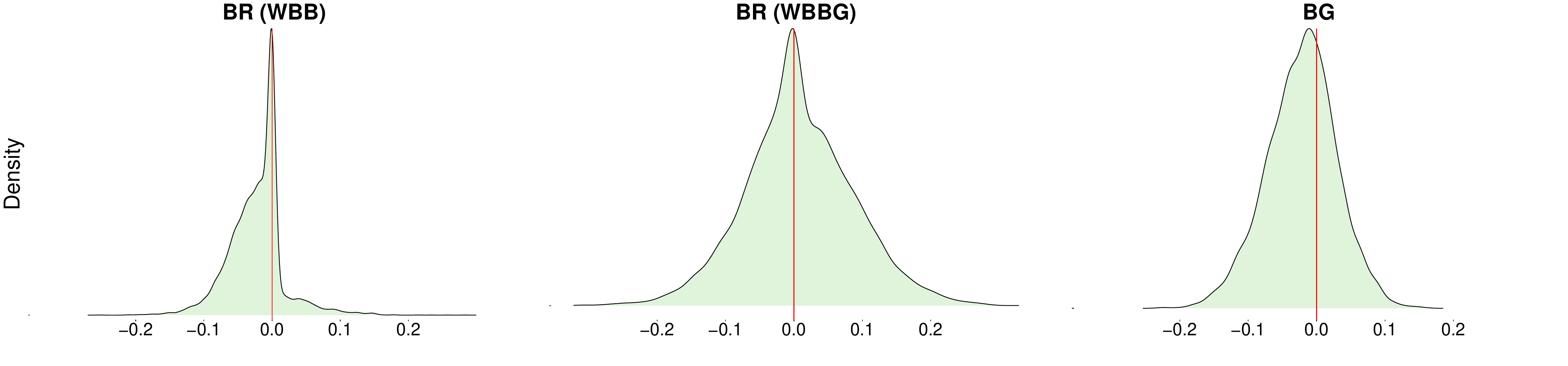}
\caption{Oneshot example of the posterior distributions. From left to right, the posteriors are BR (WBB), BR (WBBG), and BG (Gibbs sampling), respectively.}
\label{posterior-BR-SBR-BG}
\end{center}
\end{figure}

\begin{table}[htbp]

\caption{The mean of TPR, FPR, and FDR based on 100 repetitions for the data-generating process (a) and (b). }
\centering
\begin{tabular}{cc|cccccc}
\toprule
  \multicolumn{8}{c}{Data-generating process (a)} \\
  \hline
 &  & BR1 & BR2 & BR3 & BR4 & BR5 & BR (WBBG) \\ 
  \hline
 \multirow{4}{*}{(A)}  
  & RMSE & 0.03 & 0.02 & 0.02 & 0.02 & 0.02 & 0.03 \\ 
  & TPR & 0.91 & 0.82 & 0.77 & 0.73 & 0.70 & 0.96 \\ 
  & FPR & 0.89 & 0.46 & 0.29 & 0.20 & 0.14 & 0.96 \\ 
  & FDR & 0.02 & 0.04 & 0.06 & 0.07 & 0.07 & 0.01 \\ 
    \hline
  \multirow{4}{*}{(B)}
   & RMSE & 0.03 & 0.02 & 0.02 & 0.02 & 0.02 & 0.03 \\ 
  & TPR & 0.92 & 0.83 & 0.78 & 0.73 & 0.70 & 0.96 \\ 
  & FPR & 0.92 & 0.52 & 0.32 & 0.21 & 0.16 & 0.96 \\ 
  & FDR & 0.02 & 0.04 & 0.05 & 0.07 & 0.07 & 0.01 \\ 
   \midrule
   \multicolumn{8}{c}{Data-generating process (b)} \\
     \hline
 &  & BR1 & BR2 & BR3 & BR4 & BR5 & BR (WBBG) \\ 
   \hline
   \multirow{4}{*}{(A)}
  & RMSE & 0.04 & 0.04 & 0.04 & 0.04 & 0.04 & 0.05 \\ 
   & TPR & 0.92 & 0.82 & 0.78 & 0.74 & 0.71 & 0.97 \\ 
   & FPR & 0.89 & 0.49 & 0.31 & 0.23 & 0.17 & 0.97 \\ 
   & FDR & 0.02 & 0.04 & 0.05 & 0.06 & 0.07 & 0.01 \\  
     \hline
 \multirow{4}{*}{(B)} 
  & RMSE & 0.04 & 0.04 & 0.04 & 0.04 & 0.03 & 0.05 \\ 
   & TPR & 0.92 & 0.82 & 0.78 & 0.74 & 0.71 & 0.97 \\ 
   & FPR & 0.90 & 0.50 & 0.31 & 0.22 & 0.16 & 0.97 \\ 
   & FDR & 0.02 & 0.04 & 0.05 & 0.06 & 0.07 & 0.01 \\ 
   \bottomrule
\end{tabular}
\label{Sample-lambda-ab}
\end{table}

\begin{table}[htbp]

\caption{The mean of TPR, FPR, and FDR based on 100 repetitions for the data-generating process (c). }
\centering
\begin{tabular}{cc|cccccc}
\toprule
    \multicolumn{8}{c}{Data-generating process (c) $\eta=5$} \\
      \hline
 &  & BR1 & BR2 & BR3 & BR4 & BR5 & BR (WBBG) \\ 
    \hline
   \multirow{4}{*}{(A)} 
   & RMSE & 0.04 & 0.04 & 0.04 & 0.04 & 0.04 & 0.05 \\ 
   & TPR & 0.91 & 0.82 & 0.77 & 0.74 & 0.71 & 0.97 \\ 
   & FPR & 0.90 & 0.49 & 0.31 & 0.21 & 0.15 & 0.97 \\ 
   & FDR & 0.02 & 0.04 & 0.06 & 0.06 & 0.07 & 0.01 \\  
     \hline
   \multirow{4}{*}{(B)} 
   & RMSE & 0.07 & 0.10 & 0.11 & 0.13 & 0.14 & 0.08 \\ 
   & TPR & 1.00 & 1.00 & 0.98 & 0.95 & 0.91 & 1.00 \\ 
   & FPR & 0.97 & 0.40 & 0.26 & 0.19 & 0.15 & 1.00 \\ 
   & FDR & 0.00 & 0.00 & 0.01 & 0.02 & 0.04 & 0.00 \\ 
     \midrule
    \multicolumn{8}{c}{Data-generating process (c) $\eta=10$} \\
      \hline
 &  & BR1 & BR2 & BR3 & BR4 & BR5 & BR (WBBG) \\ 
    \hline
   \multirow{4}{*}{(A)}
   & RMSE & 0.07 & 0.10 & 0.11 & 0.13 & 0.14 & 0.08 \\ 
   & TPR & 1.00 & 0.99 & 0.98 & 0.94 & 0.90 & 1.00 \\ 
   & FPR & 0.99 & 0.44 & 0.28 & 0.20 & 0.15 & 1.00 \\  
   & FDR & 0.00 & 0.00 & 0.01 & 0.03 & 0.05 & 0.00 \\ 
     \hline
   \multirow{4}{*}{(B)} 
   & RMSE & 0.11 & 0.13 & 0.14 & 0.15 & 0.16 & 0.14 \\ 
   & TPR & 1.00 & 0.99 & 0.96 & 0.92 & 0.88 & 1.00 \\ 
   & FPR & 0.99 & 0.45 & 0.28 & 0.21 & 0.16 & 1.00 \\ 
   & FDR & 0.00 & 0.01 & 0.02 & 0.04 & 0.05 & 0.00 \\ 
   \midrule
    \multicolumn{8}{c}{Data-generating process (c) $\eta=20$} \\
      \hline
 &  & BR1 & BR2 & BR3 & BR4 & BR5 & BR (WBBG) \\ 
    \hline
   \multirow{4}{*}{(A)} 
   & RMSE & 0.10 & 0.12 & 0.13 & 0.15 & 0.15 & 0.15 \\ 
   & TPR & 1.00 & 0.99 & 0.97 & 0.92 & 0.88 & 1.00 \\ 
   & FPR & 0.99 & 0.44 & 0.28 & 0.21 & 0.16 & 1.00 \\ 
   & FDR & 0.00 & 0.00 & 0.02 & 0.04 & 0.05 & 0.00 \\ 
     \hline
  \multirow{4}{*}{(B)} 
  & RMSE & 0.10 & 0.12 & 0.13 & 0.15 & 0.15 & 0.15 \\ 
   & TPR & 1.00 & 0.99 & 0.97 & 0.93 & 0.88 & 1.00 \\ 
   & FPR & 0.99 & 0.44 & 0.28 & 0.20 & 0.15 & 1.00 \\ 
   & FDR & 0.00 & 0.01 & 0.02 & 0.03 & 0.05 & 0.00 \\  
   \bottomrule
\end{tabular}
\label{Sample-lambda-c}
\end{table}

\section{Comparison with the non-robust method under removing outliers}

We compare the proposed method with the (non-robust) BG method when an off-the-shelf outlier detection method is used to identify and remove the outliers before applying the BG method.
We employ Hotelling’s $T^2$ method as an outlier detection method constructed by the following steps:
\begin{enumerate}
\item Calculate sample mean vector $\hat{\bm{\mu}}=n^{-1}\sum_{i=1}^n \bm{y}_i$ and sample covariance matrix $\hat{\mathbf{\Sigma}}=n^{-1}\sum_{i=1}^n (\bm{y}_i-\hat{\bm{\mu}})(\bm{y}_i-\hat{\bm{\mu}})^{\top}$.
\item Calculate the Mahalanobis distance $g(\bm{y}_i)=(\bm{y}_i-\hat{\bm{\mu}})^{\top}\hat{\mathbf{\Sigma}}^{-1}(\bm{y}_i-\hat{\bm{\mu}})$ for each $\bm{y}_i$.
\item If $g(\bm{y}_i)>\chi_{p}(\alpha)$, then $y_i$ is regarded an outlier, where $\chi_{p}(\alpha)$ is a threshold and is approximated by the $1-\alpha$ quantile of chi squared distribution with degree of freedom $p$.
\end{enumerate}
Note that the detection method assumes that non-outliers are distributed to the multivariate Gaussian distribution. The value of $\alpha$ plays an important role in the detection level of outliers. We conventionally use $\alpha=0.01,0.05$. We adopt the simulation setting in the main manuscript to compare it with the proposed method. The result is summarized in Table~\ref{Remove-outlier}, where the result of the proposed robust method (BR) is the same as that of the main manuscript. From Table~\ref{Remove-outlier}, it seems that the removing method reasonably works as well as the BR method for scenarios (a) and (b). However, one of the proposed BR methods has performed best in scenario (c) $\eta=5, 10, 20$. Note that the proposed method may be affected by moderate outliers because the theoretical property in the proposal assumes an extremely large outlier. Although both the mean of distribution (a) and the mixture distribution (b) are zero, the mean of mixture distribution (c) is shifted. Therefore, $\hat{\bm{\mu}}$ in Hotelling's $T^2$ method is affected by outliers, and then whether Hotelling’s $T^2$ method works or not seems to depend on the structure of the outliers.

\begin{table}[htbp]

\caption{The mean of RMSE, TPR, FPR, and FDR based on 100 repetitions for all data-generating processes. }
\centering
\begin{tabular}{c|cccc|cccc}
\toprule
  \multicolumn{9}{c}{Data-generating process (a)} \\
  \hline
  &\multicolumn{4}{c|}{(A)} &  \multicolumn{4}{c}{(B)}  \\
  \hline
  & RMSE & TPR & FPR & FDR & RMSE & TPR & FPR & FDR \\ 
 \hline
BR1 & 0.026 & 0.908 & 0.885 & 0.023 & 0.072 & 1.000 & 0.970 & 0.000 \\ 
  BR2 & 0.022 & 0.818 & 0.461 & 0.045 & 0.095 & 0.997 & 0.404 & 0.001 \\ 
  BR3 & 0.021 & 0.769 & 0.291 & 0.057 & 0.114 & 0.983 & 0.262 & 0.008 \\ 
  BR4 & 0.021 & 0.732 & 0.201 & 0.066 & 0.128 & 0.952 & 0.190 & 0.022 \\ 
  BR5 & 0.022 & 0.701 & 0.143 & 0.073 & 0.139 & 0.910 & 0.150 & 0.042 \\  
  BG ($\alpha=0.01$) & 0.021 & 0.438 & 0.279 & 0.138 & 0.064 & 1.000 & 0.369 & 0.000 \\ 
  BG ($\alpha=0.05$)  & 0.023 & 0.441 & 0.294 & 0.137 & 0.067 & 1.000 & 0.376 & 0.000 \\
   \midrule
   \multicolumn{9}{c}{Data-generating process (b)} \\
  \hline
  &\multicolumn{4}{c|}{(A)} &  \multicolumn{4}{c}{(B)}  \\
  \hline
 & RMSE & TPR & FPR & FDR & RMSE & TPR & FPR & FDR \\ 
 \hline
  BR1 & 0.027 & 0.922 & 0.919 & 0.019 & 0.072 & 1.000 & 0.988 & 0.000 \\ 
  BR2 & 0.023 & 0.834 & 0.517 & 0.041 & 0.095 & 0.992 & 0.439 & 0.004 \\ 
  BR3 & 0.022 & 0.782 & 0.319 & 0.054 & 0.113 & 0.976 & 0.275 & 0.011 \\ 
  BR4 & 0.022 & 0.728 & 0.214 & 0.067 & 0.127 & 0.941 & 0.196 & 0.028 \\ 
  BR5 & 0.023 & 0.696 & 0.156 & 0.075 & 0.138 & 0.902 & 0.153 & 0.046 \\ 
  BG ($\alpha=0.01$) & 0.023 & 0.437 & 0.280 & 0.138 & 0.070 & 0.999 & 0.365 & 0.001 \\ 
  BG ($\alpha=0.05$) & 0.022 & 0.434 & 0.279 & 0.139 & 0.067 & 1.000 & 0.364 & 0.000 \\ 
  \midrule
   \multicolumn{9}{c}{Data-generating process (c) $\eta=5$} \\
  \hline
  &\multicolumn{4}{c|}{(A)} &  \multicolumn{4}{c}{(B)}  \\
  \hline
 & RMSE & TPR & FPR & FDR & RMSE & TPR & FPR & FDR \\ 
 \hline
BR1 & 0.044 & 0.916 & 0.890 & 0.021 & 0.109 & 1.000 & 0.992 & 0.000 \\ 
  BR2 & 0.040 & 0.820 & 0.489 & 0.044 & 0.129 & 0.987 & 0.448 & 0.006 \\ 
  BR3 & 0.039 & 0.778 & 0.314 & 0.055 & 0.144 & 0.961 & 0.284 & 0.018 \\ 
  BR4 & 0.037 & 0.744 & 0.226 & 0.063 & 0.155 & 0.925 & 0.205 & 0.035 \\ 
  BR5 & 0.037 & 0.711 & 0.167 & 0.071 & 0.163 & 0.882 & 0.158 & 0.055 \\ 
  BG ($\alpha=0.05$) & 0.042 & 0.361 & 0.251 & 0.157 & 0.148 & 0.836 & 0.342 & 0.077\\ 
  BG ($\alpha=0.05$) & 0.044 & 0.368 & 0.265 & 0.155 & 0.142 & 0.849 & 0.354 & 0.070 \\ 
 \midrule
   \multicolumn{9}{c}{Data-generating process (c) $\eta=10$} \\
  \hline
  &\multicolumn{4}{c|}{(A)} &  \multicolumn{4}{c}{(B)}  \\
  \hline
 & RMSE & TPR & FPR & FDR & RMSE & TPR & FPR & FDR \\ 
 \hline
 BR1 & 0.041 & 0.916 & 0.904 & 0.021 & 0.098 & 1.000 & 0.990 & 0.000 \\ 
  BR2 & 0.037 & 0.824 & 0.498 & 0.043 & 0.118 & 0.990 & 0.443 & 0.005 \\ 
  BR3 & 0.036 & 0.778 & 0.311 & 0.054 & 0.134 & 0.966 & 0.282 & 0.016 \\ 
  BR4 & 0.035 & 0.740 & 0.218 & 0.064 & 0.146 & 0.924 & 0.205 & 0.035 \\ 
  BR5 & 0.035 & 0.708 & 0.160 & 0.072 & 0.155 & 0.884 & 0.155 & 0.054 \\ 
  BG ($\alpha=0.01$) & 0.047 & 0.357 & 0.241 & 0.158 & 0.156 & 0.823 & 0.332 & 0.082 \\ 
  BG ($\alpha=0.05$) & 0.049 & 0.365 & 0.256 & 0.156 & 0.152 & 0.833 & 0.336 & 0.078 \\ 
 \midrule
   \multicolumn{9}{c}{Data-generating process (c) $\eta=20$} \\
  \hline
  &\multicolumn{4}{c|}{(A)} &  \multicolumn{4}{c}{(B)}  \\
  \hline
 & RMSE & TPR & FPR & FDR & RMSE & TPR & FPR & FDR \\ 
 \hline
 BR1 & 0.041 & 0.914 & 0.904 & 0.021 & 0.095 & 1.000 & 0.986 & 0.000 \\ 
  BR2 & 0.037 & 0.822 & 0.494 & 0.044 & 0.116 & 0.989 & 0.438 & 0.005 \\ 
  BR3 & 0.036 & 0.775 & 0.310 & 0.055 & 0.133 & 0.966 & 0.277 & 0.016 \\ 
  BR4 & 0.035 & 0.739 & 0.215 & 0.064 & 0.145 & 0.927 & 0.200 & 0.034 \\ 
  BR5 & 0.035 & 0.708 & 0.154 & 0.072 & 0.155 & 0.883 & 0.151 & 0.054 \\ 
  BG ($\alpha=0.01$) & 0.048 & 0.357 & 0.239 & 0.158 & 0.159 & 0.822 & 0.325 & 0.083 \\ 
  BG ($\alpha=0.05$) & 0.050 & 0.363 & 0.251 & 0.156 & 0.156 & 0.826 & 0.332 & 0.081 \\ 
   \bottomrule
\end{tabular}
\label{Remove-outlier}
\end{table}

\section{Additional simulation studies}

To see the effects of graph structure, we conduct additional simulation studies. In this section, the dimension of the response variables or the precision matrix $p$ is changed in the new two types of graph structure as follows: (C) the small-world graph and (D) the scale-free graph, which are randomly generated by the {\tt `sample\_smallworld'} function and the {\tt `sample\_pa'} function in the {\tt igraph} package of R software, respectively. To generate the precision matrix $\mathbf{\Omega}$ from the graph structure with an adjacency matrix $\mathbf{A}=(A_{ij})\in\mathbb{R}^{p\times p}$, we use the same method as in \cite{tan2014learning}. First, we create a matrix $\mathbf{E}=(E_{ij})\in\mathbb{R}^{p\times p}$ given by
\begin{align*}
    E_{ij}=U\left([-0.75, -0.25]\cup [0.25, 0.75]\right)1_{\{A_{ij}=1\}},
\end{align*}
where $U(D)$ is a random sample from the uniform distribution $D$.
Next, we calculate $\Tilde{\mathbf{E}}=(\mathbf{E}+\mathbf{E}^{\top})/2$ and $\mathbf{\Omega}=\Tilde{\mathbf{E}}+(0.1-\Lambda_{\mathrm{min}})\mathbf{I}$, where $\Lambda_{\mathrm{min}}$ is the smallest eigenvalue of $\Tilde{\mathbf{E}}$.
We assume $p=20$ for each structure in the additional simulation study. The visualization of graphs (C) and (D) is shown in Figure~\ref{graph-C-D}. 
The data-generating process is (a) and (b), which is the same setting as in the main manuscript. The results are shown in Tables \ref{sim-TPR-FPR-FDR} and \ref{Table-AL-CP}.
It is observed that the proposed method works well under not only (A) and (B) in the main manuscript but also (C) and (D).

\begin{figure}[htbp]
\begin{center}
\includegraphics[width=\linewidth]{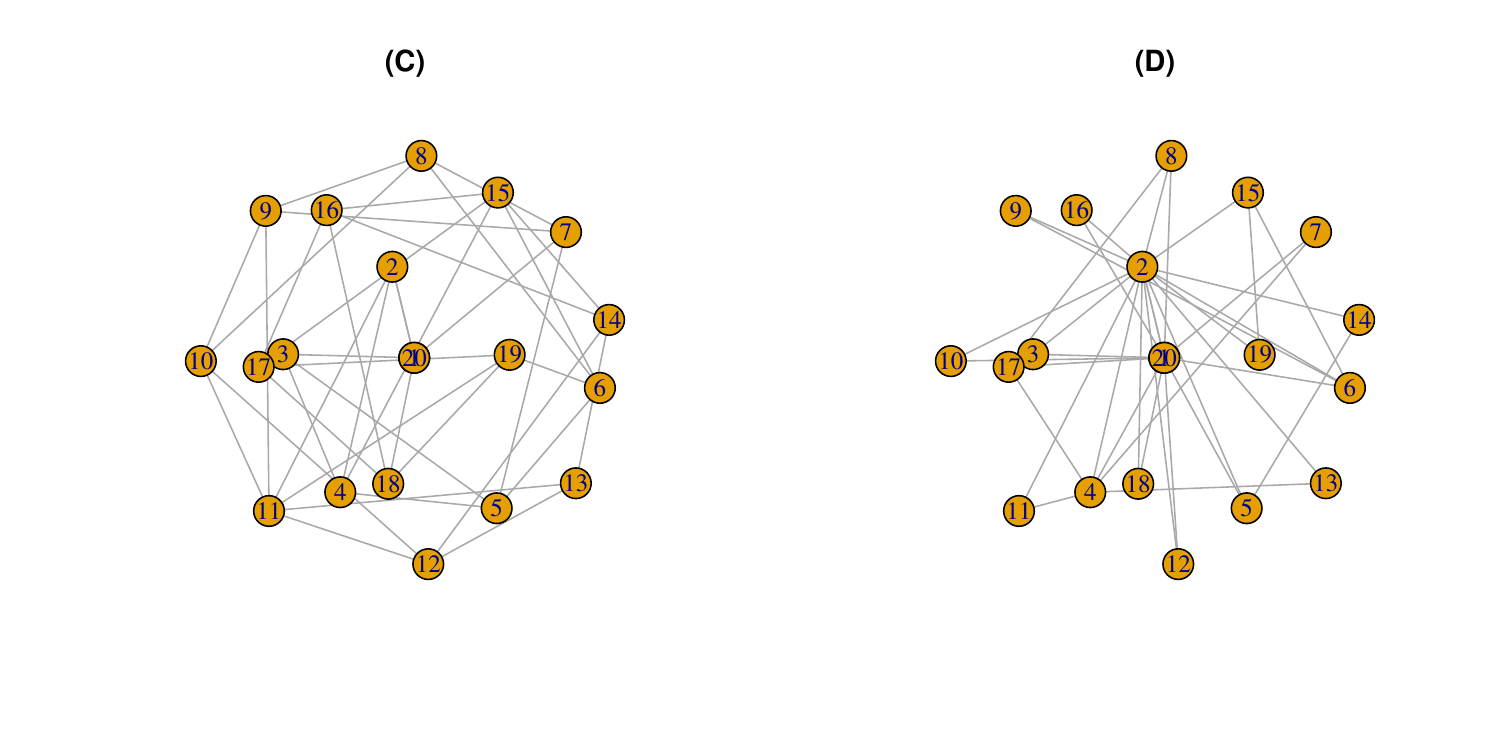}
\caption{The sampled graph structures (C) and (D).}
\label{graph-C-D}
\end{center}
\end{figure}

\begin{table}[htbp]

\caption{The mean of TPR, FPR, and FDR based on 100 repetitions for (C) and (D) with $p=20$. }
\begin{center}
\resizebox{1.0\textwidth}{!}{ 
\begin{tabular}{cc|ccccc|ccccc|ccccc|c}
  \toprule
  \multicolumn{17}{c}{Data-generating process (a)} \\
  \hline
  && BR1 & BR2 & BR3 & BR4 & BR5 & FR1 & FR2 & FR3 & FR4 & FR5 & FG1 & FG2 & FG3 & FG4 & FG5 & BG \\
  \hline
\multirow{3}{*}{(C)}& TPR & 0.90 & 0.75 & 0.7 & 0.67 & 0.65 & 0.83 & 0.75 & 0.70 & 0.67 & 0.64 & 0.85 & 0.77 & 0.72 & 0.69 & 0.66 & 0.35 \\ 
  & FPR & 0.72 & 0.25 & 0.15 & 0.11 & 0.09 & 0.46 & 0.27 & 0.20 & 0.16 & 0.14 & 0.50 & 0.30 & 0.22 & 0.18 & 0.16 & 0.42 \\ 
   & FDR & 0.03 & 0.07 & 0.08 & 0.09 & 0.09 & 0.04 & 0.07 & 0.08 & 0.09 & 0.10 & 0.04 & 0.06 & 0.07 & 0.08 & 0.09 & 0.17 \\ 
  \multirow{3}{*}{(D)} & TPR & 0.86 & 0.74 & 0.69 & 0.66 & 0.64 & 0.8 & 0.75 & 0.72 & 0.68 & 0.62 & 0.81 & 0.76 & 0.74 & 0.71 & 0.66 & 0.37 \\ 
   & FPR & 0.47 & 0.15 & 0.07 & 0.04 & 0.03 & 0.34 & 0.14 & 0.07 & 0.05 & 0.03 & 0.39 & 0.17 & 0.09 & 0.06 & 0.04 & 0.44 \\ 
   & FDR & 0.04 & 0.06 & 0.07 & 0.08 & 0.09 & 0.05 & 0.06 & 0.07 & 0.08 & 0.09 & 0.04 & 0.06 & 0.06 & 0.07 & 0.08 & 0.15 \\ 
\midrule
  \multicolumn{17}{c}{Data-generating process (b)} \\
  \hline
  && BR1 & BR2 & BR3 & BR4 & BR5 & FR1 & FR2 & FR3 & FR4 & FR5 & FG1 & FG2 & FG3 & FG4 & FG5 & BG\\
  \hline
\multirow{3}{*}{(C)}& TPR & 0.93 & 0.76 & 0.71 & 0.67 & 0.65 & 0.84 & 0.75 & 0.70 & 0.67 & 0.64 & 0.95 & 0.90 & 0.86 & 0.82 & 0.78 & 0.36 \\ 
   & FPR & 0.81 & 0.28 & 0.17 & 0.12 & 0.09 & 0.48 & 0.28 & 0.21 & 0.17 & 0.15 & 0.92 & 0.86 & 0.81 & 0.76 & 0.70 & 0.40 \\ 
   & FDR & 0.02 & 0.06 & 0.08 & 0.09 & 0.09 & 0.04 & 0.07 & 0.08 & 0.09 & 0.10 & 0.01 & 0.03 & 0.04 & 0.05 & 0.06 & 0.17 \\ 
  \multirow{3}{*}{(D)} & TPR & 0.87 & 0.73 & 0.70 & 0.66 & 0.63 & 0.80 & 0.75 & 0.72 & 0.67 & 0.62 & 0.94 & 0.88 & 0.83 & 0.79 & 0.75 & 0.39 \\ 
  & FPR & 0.52 & 0.17 & 0.08 & 0.04 & 0.03 & 0.37 & 0.15 & 0.08 & 0.05 & 0.04 & 0.91 & 0.85 & 0.79 & 0.74 & 0.69 & 0.43 \\ 
   & FDR & 0.03 & 0.06 & 0.07 & 0.08 & 0.09 & 0.05 & 0.06 & 0.07 & 0.08 & 0.09 & 0.01 & 0.03 & 0.04 & 0.05 & 0.06 & 0.15 \\
\bottomrule
\end{tabular}
}
\end{center}

\label{sim-TPR-FPR-FDR}
\end{table}

\begin{table}[htbp]
\caption{The mean of AL and CP based on 100 repetitions for (C) and (D) with $p=20$. } 
\begin{center}
\begin{tabular}{cc|ccccccc}
  \toprule
  \multicolumn{9}{c}{Data-generating process (a)} \\
  \hline
  &&BR1 & BR2 & BR3 & BR4 & BR5 &BT&BG\\
  \hline
\multirow{2}{*}{(C)} & AL & 0.383 & 0.322 & 0.285 & 0.259 & 0.238 & 0.317 & 0.318 \\ 
   & CP & 0.911 & 0.855 & 0.802 & 0.752 & 0.703 & 0.925 & 0.975 \\ 
  \multirow{2}{*}{(D)} & AL & 0.439 & 0.359 & 0.311 & 0.278 & 0.251 & 0.363 & 0.392 \\ 
   & CP & 0.884 & 0.806 & 0.732 & 0.659 & 0.593 & 0.932 & 0.977 \\ 
\midrule
\multicolumn{9}{c}{Data-generating process (b)} \\
   \hline
  &&BR1 & BR2 & BR3 & BR4 & BR5 &BT&BG\\
  \hline
\multirow{2}{*}{(C)} & AL & 0.421 & 0.351 & 0.310 & 0.28 & 0.258 & 0.204 & 0.104 \\ 
   & CP & 0.913 & 0.861 & 0.809 & 0.759 & 0.713 & 0.550 & 0.530 \\ 
  \multirow{2}{*}{(D)} & AL & 0.484 & 0.394 & 0.339 & 0.303 & 0.274 & 0.207 & 0.126 \\ 
   & CP & 0.887 & 0.812 & 0.739 & 0.672 & 0.610 & 0.536 & 0.521 \\
\bottomrule
\end{tabular}
\end{center}

\label{Table-AL-CP}
\end{table}

\section{Scalability of the proposed algorithm}

We investigate the scalability of the proposed algorithm for dimension $p$. We used the matrix (C) and (D) for $p=12, 20, 50$ and calculated the average run-time over 20 repetitions. The number of approximate posterior samples via the weighted Bayesian bootstrap method is 1000 by following \cite{newton2021weighted}.
Note that it is impossible to see the impact of dimension $p$ only because the matrix is not common for $p$.
The result is summarized in Table~\ref{table-run-time}. From Table~\ref{table-run-time}, the algorithm is computationally expensive as the dimension $p$ is large. Although the algorithm is more efficient because of the following two reasons: 1) it is not necessary to consider the autocorrelation of MCMC chains; 2) there is no accept/reject step such as the Metropolis algorithm, it is an important issue to improve the algorithm. In addition, as well as the MCMC algorithm, the number of posterior samples may need to be large to approximate the posterior if $p$ is large. Note that the computation cost depends on the optimization time and parallel computation is also possible.

\begin{table}[htbp]
\caption{The mean of run-time based on 20. } 
\begin{center}
\begin{tabular}{cc|ccc}
\toprule
&\multicolumn{4}{c}{Time (second)} \\
\midrule
 && $p=12$ & $p=20$ & $p=50$ \\ 
\midrule
\multirow{2}{*}{(a)}&(C) & 8 & 14 & 228 \\ 
  & (D) & 9 & 14 & 222 \\ 
  \midrule
\multirow{2}{*}{(b)}& (C) & 10 & 14 & 234 \\ 
  & (D) & 10 & 14 & 224 \\ 
\bottomrule
\end{tabular}
\end{center} 
\label{table-run-time}
\end{table}

\section{Additional information of real data analysis}

The real data analysis is given in Section 5 of the main manuscript. In this section, we provide additional information for the analysis.
First, we focus on the effect of the selection of $\lambda$ in real data analysis. Following \cite{hirose2017robust}, we selected $\lambda$ as the number of edges is 9 in the main manuscript. However, the solution path changes slightly for the selection of $\lambda$. The results with a different value of $\lambda$ are shown in Figure~\ref{example1-solutionpath-BR}. For example, if we use $\lambda=0.068$ instead of $\lambda=0.078$ under the data without outliers (right panels), then the path between GAL4 and GAL2 is drawn. Moreover, employing $\lambda=0.067$ instead of $\lambda=0.068$, it is observed that the path between GAL11 and GAL3 is drawn additionally.

Next, we see the result of the BG in the real data example. The BG method adaptively estimates the number of edges because $\lambda$ is sampled from the Gibbs sampler. Therefore, the number of edges for each data is different. In addition, the result using Hotelling's method differs from that without detected outliers, and then Hotelling's method does not work well in this example. The numbers of edges are 7, 5, and 9 from left to right.

\begin{figure}[htbp]
\begin{center}
\includegraphics[width=\linewidth]{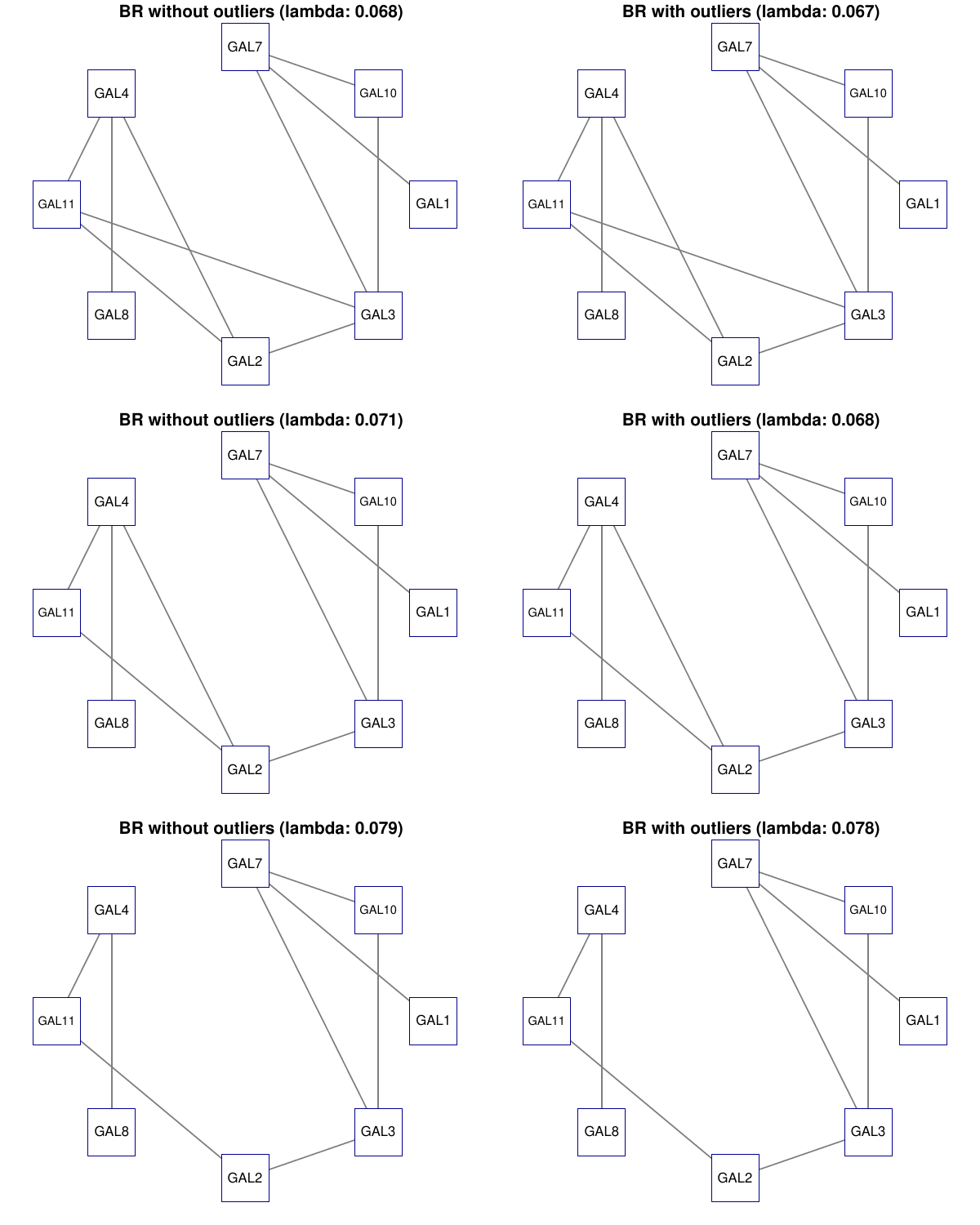}
\caption{Estimated graphical models via the BR. The number of edges is 10, 9, and 8 from top to bottom. The left panels are based on the data without 13 outliers, and the right panels are based on all data with outliers.}
\label{example1-solutionpath-BR}
\end{center}
\end{figure}

\begin{figure}[htbp]
\begin{center}
\includegraphics[width=\linewidth]{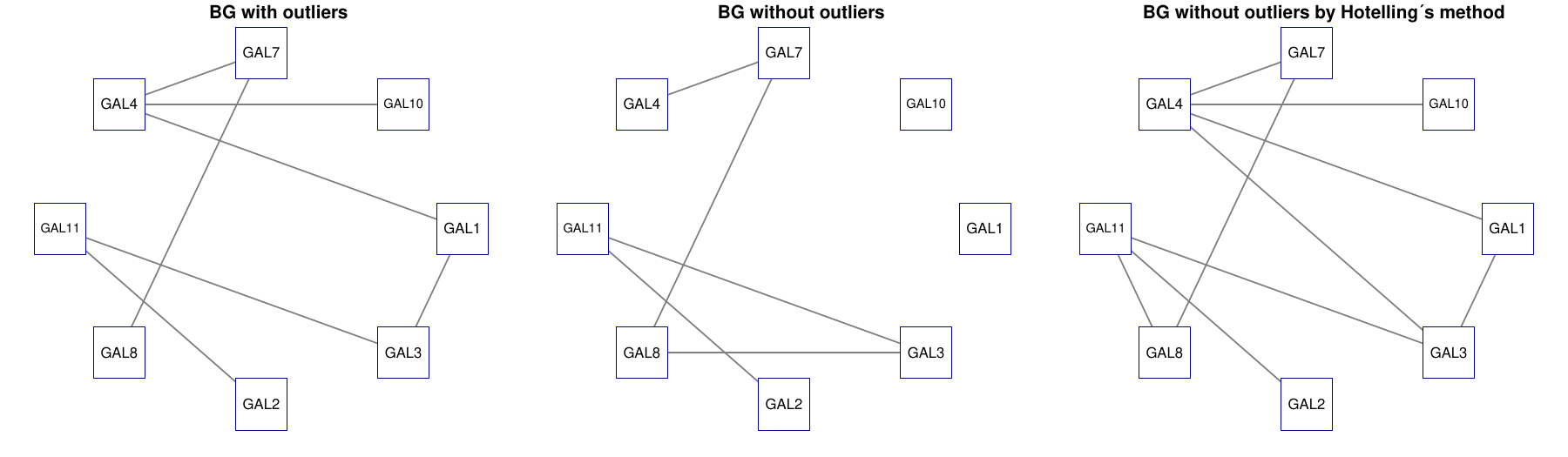}
\caption{Estimated graphical models via the BG. The right panel is based on the data removed outliers by Hotelling's method, the center panel is based on the data without 13 outliers, and the left panel is based on all data with outliers.}
\label{example1-solutionpath-BR}
\end{center}
\end{figure}

\end{document}